\documentclass[12pt]{article}
\usepackage{amsmath}
\usepackage{amsthm}
\usepackage{graphicx,psfrag,epsf}
\usepackage{enumerate}
\usepackage{natbib}
\usepackage[utf8]{inputenc} % allow utf-8 input
\usepackage{hyperref}       % hyperlinks
\usepackage{url}            % simple URL typesetting
\usepackage{booktabs}       % professional-quality tables
\usepackage{amsfonts}       % blackboard math symbols
\usepackage{nicefrac}       % compact symbols for 1/2, etc.
\usepackage{microtype}      % microtypography
\usepackage{authblk}
\usepackage{natbib}
\usepackage{doi}
\usepackage{float}
\usepackage{algorithmic}
\usepackage{xcolor}
\usepackage[toc,page]{appendix}

\newtheorem{theorem}{Theorem}[section]
 
\newtheorem{lemma}[theorem]{Lemma}

\newtheorem{assumption}[theorem]{Assumption}

\begin{document}

%%%%%%%%%%%%%%%%%%%%%%%%%%%%%%%%%%%%%%%%%%%%%%%%%%%%%%%%%%%%%%%%%%%%%%%%%%%%%%

  \title{Scalable $k$-d trees for distributed data}
  \author{Aritra Chakravorty, William S. Cleveland, and Patrick J. Wolfe}
  \affil{\texttt{\{chakrav0,wsc,patrick\}@purdue.edu}}
\maketitle

\begin{abstract}
Data structures known as $k$-d trees have numerous applications in scientific computing, particularly in areas of modern statistics and data science such as range search in decision trees, clustering, nearest neighbors search, local regression, and so forth. In this article we present a scalable mechanism to construct $k$-d trees for distributed data, based on approximating medians for each recursive subdivision of the data. We provide theoretical guarantees of the quality of approximation using this approach, along with a simulation study quantifying the accuracy and scalability of our proposed approach in practice.
\end{abstract}

\noindent%

\vspace{\baselineskip}%
\noindent%
\section{Introduction}\label{sec:introduction}

Local neighborhood queries, as a form of proximity search, are of fundamental interest in modern scientific computing for large data sets. Indeed, many areas of statistics and data science such as range search in a statistical decision tree, clustering, nearest neighbors search, local regression, and so on make use of such ideas. However, when data sets are stored in a contemporary distributed environment such as Spark \citep{Zaharia2016}, identifying the local neighborhood around a given point in a multidimensional setting can be challenging: There is less scope of interaction between subsets of a distributed data set, and often it becomes increasingly difficult to implement more mathematically or computationally sophisticated algorithms. A significant current literature aims to address this problem. For example, low-distortion embedding \citep{Ailon2009} can speed up search algorithms in approximate nearest neighbors, whereas parallel algorithms for nearest neighbor search \citep{Hu2015,Chen2018,Pinkman2020} can offer increased scalability, even potentially in high-dimensional settings \citep{Xiao2016}.

The majority of such approaches involve the construction of a $k$-d tree for multivariate data as a generic data structure \citep{Bentley1975}. Similarly to a binary tree, a $k$-d tree divides the data recursively at each level of the tree along an axis called the key for that level. However, as opposed to a standard binary tree which has only one key for every level of the tree, in a $k$-d tree uses $k$ keys and cycles through these keys for successive levels of the tree. For example, to build a $k$-d tree from three-dimensional points comprising $(x, y, z)$ coordinates, keys would by default be cycled as $x, y, z, x, y, z, \dots$, for successive levels of the $k$-d tree. Every non-leaf node divides the space into two parts, known as half-spaces. This node acts as a boundary point for the two half-spaces. Points to the left of this boundary along the key coordinate are represented by the left sub-tree of that node, and points right are represented by the right sub-tree. So, for example, if at a given level the $x$ coordinate is chosen as key, then all points in the sub-tree with a smaller $x$ value than the node will appear in the left sub-tree, and all points with larger $x$ value will be in the right sub-tree.

For this reason, $k$-d trees have seen wide use as efficient tools to subdivide point clouds of input data into sub-spaces of nearly equal volume. There is a more logically efficient but computationally complex scheme for cycling the keys that chooses the coordinate that has the widest dispersion or largest variance to be the key for a given level \citep{Friedman1977}, with the possibility of non-unique keys. Observe that, since the keys are different at different levels of a $k$-d tree, it is impossible to perform any re-balancing techniques, such as are used to build so-called AVL trees \citep{Velskii1962} or red--black trees \citep{Bayer1972, Guibas1978}. 

\begin{figure}[!htbp]
\begin{center}
\includegraphics[width=5in]{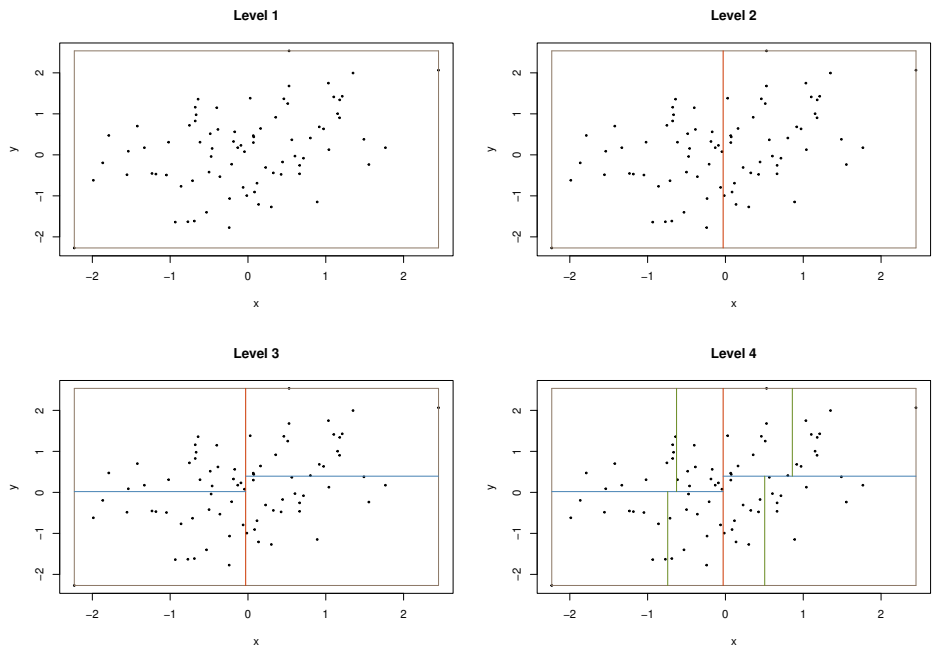}
\end{center}
\caption{\label{fig:kdtree}Canonical $k$-d tree construction. Initially at level $1$, all data points lie within the rectangular cell shown. The first splitting hyperplane is determined by the median of the $x$ coordinates of all points; it produces two new cells. At level $2$, we look at all points in the left (respectively\ right) cell, with the splitting hyperplane now made at the median of the $y$ coordinates of the points in this cell. At level $3$, again we choose the corresponding $x$ coordinates to make the split. Median computation means that cell counts are almost equal across each level, ensuring that the $k$-d tree is balanced.}
\end{figure}

As shown in Figure \ref{fig:kdtree}, the canonical approach to constructing a $k$-d tree computes medians of subsets of the input data at each level of the tree. As we move down the levels, we choose data variables periodically to select splitting hyperplanes. We split along the median of the selected variable. Choosing the medians as splitting hyperplanes constructs a balanced $k$-d tree, whereby each leaf node is approximately the same distance from the root. If medians can be found in time $O(n)$, it is possible to construct a $k$-d tree in time $O(n\log n)$ \citep{Bentley1975}. This cannot be done with exact sorting algorithms such as quick-sort, merge-sort, or heap-sort, with the latter leading instead to $O(n\log^2n)$ time for $k$-d tree construction \citep{Wald2006}. Note that the use of medians also means that the entire input data set must be read upon algorithm initialization.

Our main interest here lies developing in a parallel-computation approach to approximating medians, which also has the advantage of being easily implemented using typical programming models for contemporary distributed computing environments (e.g., map-reduce \citep{Dean2008}). Other authors have recently explored a batched incremental approach of constructing $k$-d trees \citep{Bleloch2018}, an adaptive split and sort strategy for parallel $k$-d trees construction \citep{Wehr2018}, and a construction based on pre-sorted results \citep{Brown2015b, Cao2020}. Other pre-sorting algorithms are common---to avoid re-sorting at each level of the tree---but are known to have poor worst-case performance despite a best-case complexity of $O(n\log n)$ \citep{Wald2006}. Another approach is to pre-sort points in each of $k$ dimensions, and then maintain the order of these $k$ sorts when building a $k$-d tree, achieving a worst-case complexity of $O(kn\log n)$ \citep{Procopiuc2003, Brown2015}. Additional constructions predominate in the ray-tracing literature for computer graphics \citep{Hunt2006, Shevtsov2007, Zhou2008, Soupikov2008, Choi2010}.

This article introduces a scalable parallel algorithm to construct balanced $k$-d trees through the approximation of medians. We first describe the intuition underlying our approach in section \ref{sec:intution}, and then provide our main theoretical results in section \ref{sec:main}. Next, in section \ref{sec:algorithm}, we develop an algorithm for balanced approximate $k$-d tree construction based on these results, and in section \ref{sec:performance} we describe a simulation study quantifying accuracy and scalability. Finally, we conclude in section \ref{sec:disc} with a brief discussion.
 
\section{Intuition underlying $k$-d tree construction in parallel, distributed environments}\label{sec:intution}

Assume that we have a data set $X$, possibly large and distributed, and an index set $I$ such that the $i$th element of $X$ is denoted as $x_{\:i}$ for $i\in I$. We first describe the $k$-d tree construction procedure in terms of boundary points or vertices and neighborhood or cells. Let us consider the usual construction algorithm of a $k$-d tree. It is a recursive algorithm whereby in each recursion, we get a new level or set of nodes from the previous level or set of nodes. These nodes can be also thought in terms of vertices of neighborhoods. Suppose we are at a certain level $d$ of the $k$-d tree, and we want to get the vertices for the next level. We have a number of neighborhoods in this level, and if we get the new vertices corresponding to each of these neighborhoods, we can easily get all the sub-neighborhoods for the next level $d+1$. 
First, let us figure out how to get the vertex where we divide a neighborhood. Let a general neighborhood at depth $d$ be denoted as: $\big(\mathbf{a},\mathbf{b}\big)=\big((a_{\:1},b_{\:1}),\dots,(a_{\:p},b_{\:p})\big)$. Suppose at depth $d$ that we want divide the data along the $t$th coordinate. Then, we must compute the median $m_{\:t}(\mathbf{X},\mathbf{a},\mathbf{b})$ along this $t$th coordinate for all observations $\mathbf{x}_{\:i}$ that lie inside the interval $(\mathbf{a}, \mathbf{b})$. Let us denote the total number of observations $\mathbf{x}_{\:i}$ that lie inside the open interval $(\mathbf{a},\mathbf{b})$ by $\boldsymbol{\mathsf{1}}(\mathbf{X},\mathbf{a},\mathbf{b})$. Observe that we have the expression
\begin{align*}
\boldsymbol{\mathsf{1}}(\mathbf{X},\mathbf{a},\mathbf{b})=\sum_{i\in I}\mathsf{1}(\mathbf{a}<\mathbf{x}_{\:i}<\mathbf{b})=\sum_{i\in I}\left(\prod_{k=1}^p\mathsf{1}(a_{\:k}<x_{\:i,\:k}<b_{\:k})\right).
\end{align*}

For a general $(\mathbf{a},\mathbf{b})\subseteq(\mathbf{0},\mathbf{1})$ and $m\in\mathbb{R}$, let us define the $(p-1)$-dimensional set:
\begin{align*}
\mathcal{A}_{\:t}(m,\mathbf{a},\mathbf{b})=\big\{\mathbf{x}\in(\mathbf{0},\mathbf{1})\colon x_{\:t}=m,a_{\:k}<x_{\:k}<b_{\:k}\text{ for }1\leq k\leq p,k\neq t\big\}. 
\end{align*}
Let us also introduce $p$-dimensional vector $\mathbf{a}_{\:\text{-}t}(m)$ such that:
\begin{align*}
\tilde{\mathbf{a}}=\mathbf{a}_{\:\text{-}t}(m)\leftrightarrow\tilde{a}_{\:i}=a_{\:i}\text{ for }i\neq t, 1\leq i\leq p\text{ and }\tilde{a}_{\:t}=m.
\end{align*}
Observe $(\mathbf{a},\mathbf{b})$ is the disjoint union of $\big(\mathbf{a},\mathbf{b}_{\:\text{-}t}(m)\big)$, $\big(\mathbf{a}_{\:\text{-}t}(m),\mathbf{b}\big)$ and $\mathcal{A}_{\:t}(m,\mathbf{a},\mathbf{b})$. Since $m_{\:t}(\mathbf{X},\mathbf{a},\mathbf{b})$ the median along the $t$th coordinate for all the observations $\mathbf{x}_{\:i}$, that lie inside the interval $(\mathbf{a}, \mathbf{b})$, the two disjoint intervals: $\big(\mathbf{a},\mathbf{b}_{\:\text{-}t}\big(m_{\:t}(\mathbf{X},\mathbf{a},\mathbf{b})\big)\big)$ and $\big(\mathbf{a}_{\:\text{-}t}\big(m_{\:t}(\mathbf{X},\mathbf{a},\mathbf{b})\big),\mathbf{b}\big)$ should contain equal number of observations. 

Therefore, we can find the median $m_{\:t}(\mathbf{X},\mathbf{a},\mathbf{b})$ by solving the following equation for $m\in(a_{\:t},b_{\:t})$:
\begin{align*}
\boldsymbol{\mathsf{1}}\big(\mathbf{X},\mathbf{a},\mathbf{b}_{\:\text{-}t}(m)\big)=\boldsymbol{\mathsf{1}}\big(\mathbf{X},\mathbf{a}_{\:\text{-}t}(m),\mathbf{b}\big).
\end{align*}
We standardize the function $\boldsymbol{\mathsf{1}}(\mathbf{X},\mathbf{a},\mathbf{b})$ to $\bar{\boldsymbol{\mathsf{1}}}(\mathbf{X},\mathbf{a},\mathbf{b})=(1/|\mathbf{X}|)\cdot\boldsymbol{\mathsf{1}}(\mathbf{X},\mathbf{a},\mathbf{b})$. Observe that $m_{\:t}(\mathbf{X},\mathbf{a},\mathbf{b})$ is also a solution to the following equation for $m\in(a_{\:t},b_{\:t})$:
\begin{align*}
\bar{\boldsymbol{\mathsf{1}}}\big(\mathbf{X},\mathbf{a},\mathbf{b}_{\:\text{-}t}(m)\big)=\bar{\boldsymbol{\mathsf{1}}}\big(\mathbf{X},\mathbf{a}_{\:\text{-}t}(m),\mathbf{b}\big).
\end{align*}
Note that $\bar{\boldsymbol{\mathsf{1}}}(\mathbf{X},\mathbf{a},\mathbf{b})$ always lies inside $(0,1)$, as opposed to $\boldsymbol{\mathsf{1}}(\mathbf{X},\mathbf{a}, \mathbf{b})$, which almost surely grows with $|\mathbf{X}|$. With this split of $\left(\mathbf{a},\mathbf{b}\right)$ at the median $m(\mathbf{X}, \mathbf{a},\mathbf{b})$, we get sub-neighborhoods $\big(\mathbf{a}_{\:\text{left}},\mathbf{b}_{\:\text{left}}\big)=\big(\mathbf{a},\mathbf{b}_{\:\text{-}t}\big(m_{\:t}(\mathbf{X},\mathbf{a},\mathbf{b})\big)\big)$ and $\big(\mathbf{a}_{\:\text{right}},\mathbf{b}_{\:\text{right}}\big)=\big(\mathbf{a}_{\:\text{-}t}\big(m_{\:t}(\mathbf{X},\mathbf{a},\mathbf{b})\big),\mathbf{b}\big)$.

Now, we will show how the canonical $k$-d tree construction procedure can be identified in terms of the conditional medians $m_{\:t}(\mathbf{X},\mathbf{a},\mathbf{b})$. Recall our assumption that each observation of the data $\mathbf{X}$ has been scaled, so that they lie in the $p$-dimensional interval $(\mathbf{0},\mathbf{1})$, so that $0<x_{\:i,\:k}<1$ for $i\in I, 1\leq k\leq p$. We start our $k$-d tree construction algorithm with the single neighborhood $(\mathbf{0},\mathbf{1})=\big((0,1),(0,1),\dots,(0,1)\big)$. 

At level $1$, we pick the first coordinate $x_{\:1}$ as our splitting direction and split along $x_{\:1}$ at the median value $m_{\:1,\:1}(\mathbf{X})=m_{\:1}(\mathbf{X},\mathbf{0},\mathbf{1})$. This split produces two disjoint sub-neighborhoods of $(\mathbf{0},\mathbf{1})$: $\big(\big(0,m_{\:1,\:1}(\mathbf{X})\big),(0,1),\dots,(0,1)\big)=\big(\mathbf{0},\mathbf{1}_{\:\text{-}1}\big(m_{\:1,\:1}(\mathbf{X})\big)\big)$ and $\big(\big(m_{\:1,\:1}(\mathbf{X}),1\big),(0,1),\dots,(0,1)\big)=\big(\mathbf{0}_{\:\text{-}1}\big(m_{\:1,\:1}(\mathbf{X})\big),\mathbf{1}\big)$. The union of these sub-neighborhoods along with the $(p-1)$-dimensional plane $\mathcal{A}_{\:1}\big(m_{\:1,\:1}(\mathbf{X}),\mathbf{0},\mathbf{1}\big)$, is $\big(\mathbf{0},\mathbf{1}\big)$. 

At level $2$, we pick the second coordinate $x_{\:2}$ as our splitting direction and split each of these neighborhoods along $x_{\:2}$ at corresponding median values. Let us denote these two medians as $m_{\:2,\:1}(\mathbf{X})$ and $m_{\:2,\:2}(\mathbf{X})$. Then our resulting four subset neighborhoods after the second split are: $\big(\big(0,m_{\:1,\:1}(\mathbf{X})\big),\big(0,m_{\:2,\:1}(\mathbf{X})\big),\dots,(0,1)\big)$, $\big(\big(0,m_{\:1,\:1}(\mathbf{X})\big),\big(m_{\:2,\:1}(\mathbf{X}),1\big),\dots,(0,1)\big)$, $\big(\big(m_{\:1,\:1}(\mathbf{X}),1\big),\big(0,m_{\:2,\:2}(\mathbf{X})\big),\dots,(0,1)\big)$ and $\big(\big(m_{\:1,\:1}(\mathbf{X}),1\big),\big(m_{\:2,\:2}(\mathbf{X}),1\big),\dots,(0,1)\big)$. The union of these intervals, along with the $(p-1)$-dimensional planes $\mathcal{A}_{\:1}\big(m_{\:1,\:1}(\mathbf{X}),\mathbf{0},\mathbf{1}\big)$, $\mathcal{A}_{\:2}\big(m_{\:2,\:1}(\mathbf{X}),\mathbf{0},\mathbf{1}_{\:\text{-}1}\big(m_{\:1,\:1}(\mathbf{X})\big)\big)$ and $\mathcal{A}_{\:2}\big(m_{\:2,\:2}(\mathbf{X}),\mathbf{0}_{\:\text{-}1}\big(m_{\:1,\:1}(\mathbf{X})\big),\mathbf{1}\big)$, is easily verified to be $(\mathbf{0},\mathbf{1})$. 

In general we cycle through all the coordinates in this manner. Suppose at level $d$ we pick the coordinate $x_{\:t}$; then we start with $2^{d-1}$ number of neighborhoods and we split each of these neighborhoods in two subset neighborhoods along $x_{\:t}$ at the median values. Then, at level $d$ we get $2^{d-1}$ median values $m_{\:d,\:1}(\mathbf{X}),\dots,m_{\:d,\:2^{d-1}}(\mathbf{X})$. We will have $2^{d-1}$ splits at level $d$, producing $2^d$ neighborhoods before proceeding to level $d+1$. If we continue and construct a $k$-d tree of depth $D$, we will compute $2^D-1$ conditional median statistics $m_{\:1,\:1}(\mathbf{X})$; $m_{\:2,\:1}(\mathbf{X})$, $m_{\:2,\:2}(\mathbf{X})$; $\dots$; $m_{\:D,\:1}(\mathbf{X})$, $\dots,m_{\:D,\:2^{D-1}}(\mathbf{X})$ in the construction process. Observe that these $2^D-1$ statistics give us all necessary and sufficient information to construct the entire $k$-d tree. 

\section{Main results}\label{sec:main}

The $2^D-1$ conditional medians needed to construct a canonical $k$-d tree of depth $D$ cannot be computed exactly in parallel. Instead, using the techniques developed by \citet{Chakravorty2021}, we propose to approximate each such median $m_{\:d,\:k}(\mathbf{X})$ by the $J$th term of a sequence of statistics $\{\hat{m}_{\:d,\:k}^J(\mathbf{X})\}_{J=0}^\infty$ that are easily computed in parallel. Below we show $\hat{m}_{\:d,\:k}^J(\mathbf{X})$ to be the minimizer of an objective function which approximates the standardized sum-of-product of indicators $\bar{\boldsymbol{\mathsf{1}}}(\mathbf{X},\mathbf{a},\mathbf{b})$ described in section \ref{sec:intution}. Specifically, this approximation takes the form of a $J$th partial sum of a convergent basis expansion of the indicator functions, from which it can be seen that the choice of $J$ will provide a user-controlled trade-off between speed and accuracy.

\subsection{Approximation of indicator functions}\label{subsec:approx}

Let $a\in[0,1)$ and $b\in(0,1]$, such that $a<b$ and $(a,b)\subseteq(0,1)$. Define the set:
\begin{align*}
P_{\:(a,b)}=(0,1)\setminus\Big(\{a\}^{\mathsf{1}(a>0)}\cup\{b\}^{\mathsf{1}(b<1)}\Big).
\end{align*}

We begin with the following lemma.

\begin{lemma}\label{lemma:1d:exp}
For $x\in P_{\:(a,b)}$, we have the following expansion:
\begin{align*}
\mathsf{1}(a<x<b)=\sum_{j=0}^{\infty}c_{\:j}(x)g_{\:j}(a,b),
\end{align*}
where for $x\in(0,1)$ and $j\in\mathbb{N}^+$ we define
\begin{align*}
&\mathsf{c}_{\:0}(x)=1,\quad \mathsf{c}_{\:2j-1}(x)=\cos\big((2j-1)x\big),\quad \mathsf{c}_{\:2j}(x)=\sin\big((2j-1)x\big);\text{ and } \\
&\mathsf{g}_{\:0}(a,b)=1-\frac{1}{2}\big(\mathsf{1}(a>0)+\mathsf{1}(b<1)\big), \\
&\mathsf{g}_{\:2j-1}(a,b)=\frac{2}{\pi(2j-1)}\Big(\mathsf{1}(b<1)\sin\big((2j-1)b\big)-\mathsf{1}(a>0)\sin\big((2j-1)a\big)\Big), \\ 
&\mathsf{g}_{\:2j}(a,b)=\frac{2}{\pi(2j-1)}\Big(\mathsf{1}(a>0)\cos\big((2j-1)a\big)-\mathsf{1}(b<1)\cos\big((2j-1)b\big)\Big).
\end{align*}
\end{lemma}

\begin{proof}
See Appendix \ref{ap:approx_proofs}.
\end{proof}

Next, let $\mathsf{1}_{\:J}(x,a,b)$ denote the $(2J)$th partial sum in the series expansion of $\mathsf{1}(a<x<b)$ in Lemma \ref{lemma:1d:exp} for $J\in\mathbb{N}^+$, so that:
\begin{align*}
\mathsf{1}_{\:J}(x,a,b)=\sum_{j=0}^{2J}\mathsf{c}_{\:j}(x)\cdot\mathsf{g}_{\:j}(a,b)\text{ for }J\in\mathbb{N}^+.
\end{align*}
Now, for a $\delta\in(0,1)$, let us define the set:
\begin{align*}
U_{\:\delta,\:(a,b)}=(0,1)\setminus\big((a-\delta,a+\delta)^{\mathsf{1}(a>0)}\cup(b-\delta,b+\delta)^{\mathsf{1}(b<1)}\big).
\end{align*}

We then have the following lemma.

\begin{lemma}\label{lemma:1d:unif}
(A)\ For any $\delta\in(0,1)$, the sequence of functions $\{\mathsf{1}_{\:J}(x,a,b)\}_{\:J=0}^{\infty}$ uniformly converges to the limit $\mathsf{1}(a<x<b)$ for $x\in U_{\:\delta,\:(a,b)}$.

(B)\ The sequence of functions $\{\mathsf{1}_{\:J}(x,a,b)\}_{\:J=0}^{\infty}$ is uniformly bounded for $x\in(0,1)$.
\end{lemma}

\begin{proof}
See Appendix \ref{ap:approx_proofs}.
\end{proof}

Finally, for $\mathbf{a}\in[\mathbf{0},\mathbf{1}),\mathbf{b}\in(\mathbf{0},\mathbf{1}]$, we define sets
\begin{align*}
P_{\:(\mathbf{a},\mathbf{b})}=\times_{l=1}^pP_{\:(a_{\:l},b_{\:l})},\quad U_{\:\delta,\:(\mathbf{a},\mathbf{b})}=\times_{l=1}^pU_{\:\delta,\:(a_{\:l},b_{\:l})}.
\end{align*}
Let $\mathbb{N}_{\:J}=\{0,1,\dots,2J\}$, and let $\boldsymbol{j}=(j_{\:1},\dots,j_{\:p})$ denote a general index element of $\mathbb{N}_{\:J,\:p}=\mathbb{N}_{\:J}\times\dots\times\mathbb{N}_{\:J}$($p$ times). Now for each $\boldsymbol{j}\in\mathbb{N}_{\:J,\:p}$, we define the following functions:
\begin{align*}
\boldsymbol{\mathsf{c}}_{\:\boldsymbol{j}}(\mathbf{x})=\prod_{l=1}^p\mathsf{c}_{\:j_{\:l}}(x_{\:l}),\quad\boldsymbol{\mathsf{g}}_{\:\boldsymbol{j}}(\mathbf{a},\mathbf{b})=\prod_{l=1}^p\mathsf{g}_{j_{\:l}}(a_{\:l},b_{\:l}).
\end{align*}
We approximate $\mathsf{1}(\mathbf{a}<\mathbf{x}<\mathbf{b})$ with $\mathsf{1}_{\:J}(\mathbf{x},\mathbf{a},\mathbf{b})$, where
\begin{align*}
\mathsf{1}_{\:J}(\mathbf{x},\mathbf{a},\mathbf{b})&=\prod_{l=1}^p\mathsf{1}_{\:J}(x_{\:l},a_{\:l},b_{\:l})=\prod_{l=1}^p\Big(\sum_{j=0}^{2J}\mathsf{c}_{\:j}(x_{\:l})\cdot\mathsf{g}_{\:j}(a_{\:l},b_{\:l})\Big) \\
&=\sum_{\boldsymbol{j}\in\mathbb{N}_{\:J,\:p}}\boldsymbol{\mathsf{c}}_{\:\boldsymbol{j}}(\mathbf{x})\cdot\boldsymbol{\mathsf{g}}_{\:\boldsymbol{j}}(\mathbf{a},\mathbf{b}).
\end{align*}

Lastly, we have the following lemma.

\begin{lemma}\label{lemma:pd:unif}
(A)\ For any $\delta\in(0,1)$, the sequence of functions $\{\mathsf{1}_{\:J}(\mathbf{x},\mathbf{a},\mathbf{b})\}_{\:J=0}^{\infty}$ uniformly converges to the limit $\mathsf{1}(\mathbf{a}<\mathbf{x}<\mathbf{b})$ for $\mathbf{x}\in U_{\:\delta,\:(\mathbf{a},\mathbf{b})}$.

(B)\ The sequence of functions $\{\mathsf{1}_{\:J}(\mathbf{x},\mathbf{a},\mathbf{b})\}_{\:J=0}^{\infty}$ is uniformly bounded for $\mathbf{x}\in(\mathbf{0},\mathbf{1})$.
\end{lemma}

\begin{proof}
See Appendix \ref{ap:approx_proofs}.
\end{proof}

\subsection{Stochastic bounds on accuracy}\label{subsec:stoch}

Consider a probability triple $(\Omega,\mathcal{F},P)$ giving rise to independent and identically distributed realizations $\{\mathbf{x}_{\:i}, i \in I\}$ of a random variable $\tilde{\mathbf{x}}\colon\Omega\to(\mathbf{0},\mathbf{1})$, so that $\mathbf{x}_{\:i}(w)$ is the value corresponding to a sample point $w\in\Omega$ for $i\in I$. Let us fix a $w\in\Omega$ and $n\in\mathbb{N}$, and let $\mathbf{X}(w,n)$ denote the input data  $\big(\mathbf{x}_{\:1}(w),\mathbf{x}_{\:2}(w),\dots,\mathbf{x}_n(w)\big)$, so that $|\mathbf{X}|=|\mathbf{X}(w,n)|=n$. For $\boldsymbol{j}\in\mathbb{N}_{\:J\:,p}$, let us define the statistic $\boldsymbol{\mathsf{C}}_{\:\boldsymbol{j}}(\mathbf{X})=\sum_{i\in I}\boldsymbol{\mathsf{c}}_{\:\boldsymbol{j}}(\mathbf{x}_{\:i})$, and its corresponding standardized version $\bar{\boldsymbol{\mathsf{C}}}_{\:\boldsymbol{j}}(\mathbf{X})=\big(\nicefrac{1}{|\mathbf{X}|}\big)\cdot\boldsymbol{\mathsf{C}}_{\:\boldsymbol{j}}(\mathbf{X})$. Note that for any arbitrary partition $\big\{\mathbf{X}_{\:1},\dots,\mathbf{X}_{\:R}\big\}$ of the data $\mathbf{X}$ into $R$ subsets (here $R\in\mathbb{N}^+$ and $1\leq R\leq n$), we always have  $\boldsymbol{\mathsf{C}}_{\:\boldsymbol{j}}(\mathbf{X})=\sum_{r=1}^R\boldsymbol{\mathsf{C}}_{\:\boldsymbol{j}}(\mathbf{X}_{\:r})$ for $\boldsymbol{j}\in\mathbb{N}_{\:J\:,p}$.
So, for any distributed data set $\mathbf{X}$, the collection of statistics  $\boldsymbol{\mathsf{C}}_{\:\boldsymbol{j}}(\mathbf{X})$ for $\boldsymbol{j}\in\mathbb{N}_{\:J\:,p}$, can exactly be computed in parallel. Observe: $\boldsymbol{\mathsf{C}}_{\:\boldsymbol{0}}(\mathbf{X})=\sum_{i\in I}1=|\mathbf{X}|$, where $\boldsymbol{0}$ is a $p$-dimensional vector of zeros. For $\boldsymbol{j}\in\mathbb{N}_{\:J,\:p}$, we have: $\bar{\boldsymbol{\mathsf{C}}}_{\:\boldsymbol{j}}(\mathbf{X})=\big(\nicefrac{1}{\mathbf{X}}\big)\cdot\boldsymbol{\mathsf{C}}_{\:\boldsymbol{j}}(\mathbf{X})=\boldsymbol{\mathsf{C}}_{\:\boldsymbol{j}}(\mathbf{X})/\boldsymbol{\mathsf{C}}_{\:\boldsymbol{0}}(\mathbf{X})$. Thus, the collection of standardized statistics $\bar{\boldsymbol{\mathsf{C}}}_{\:\boldsymbol{j}}(\mathbf{X})$ for $\boldsymbol{j}\in\mathbb{N}_{\:J\:,p}$ also can be exactly computed in parallel.

Throughout this section, we will assume that $(\mathbf{a},\mathbf{b})$ is a fixed neighborhood inside $(\mathbf{0},\mathbf{1})$. We approximate $\boldsymbol{\mathsf{1}}(\mathbf{X},\mathbf{a},\mathbf{b})$ by $\boldsymbol{\mathsf{1}}_{\:J}(\mathbf{X},\mathbf{a},\mathbf{b})$, defined as
\begin{align*}
\boldsymbol{\mathsf{1}}_{\:J}(\mathbf{X},\mathbf{a},\mathbf{b})=\sum_{i\in I}\mathsf{1}_J(\mathbf{x}_{\:i},\mathbf{a},\mathbf{b})=\sum_{i\in I}\sum_{\boldsymbol{j}\in\mathbb{N}_{\:J\:,p}}\boldsymbol{\mathsf{c}}_{\:\boldsymbol{j}}(\mathbf{x})\cdot\boldsymbol{\mathsf{g}}_{\:\boldsymbol{j}}(\mathbf{a},\mathbf{b}).
\end{align*}
This is a finite sum and we can exchange summation to obtain
\begin{align*}
\boldsymbol{\mathsf{1}}_{\:J}(\mathbf{X},\mathbf{a},\mathbf{b})=\sum_{\boldsymbol{j}\in\mathbb{N}_{\:J\:,p}}\Big(\sum_{i\in I}\boldsymbol{\mathsf{c}}_{\:\boldsymbol{j}}(\mathbf{x}_{\:i})\Big)\cdot\boldsymbol{\mathsf{g}}_{\:\boldsymbol{j}}(\mathbf{a},\mathbf{b})=\sum_{\boldsymbol{j}\in\mathbb{N}_{\:J\:,p}}\boldsymbol{\mathsf{C}}_{\:\boldsymbol{j}}(\mathbf{X})\cdot\boldsymbol{\mathsf{g}}_{\:\boldsymbol{j}}(\mathbf{a},\mathbf{b}).
\end{align*}
If we standardize by dividing both sides by $|\mathbf{X}|$, we obtain:
\begin{align}\label{kd:eq:a}
\bar{\boldsymbol{\mathsf{1}}}_{\:J}(\mathbf{X},\mathbf{a},\mathbf{b})=\frac{1}{|\mathbf{X}|}\boldsymbol{\mathsf{1}}_{\:J}(\mathbf{X},\mathbf{a},\mathbf{b})=\sum_{\boldsymbol{j}\in\mathbb{N}_{\:J\:,p}}\bar{\boldsymbol{\mathsf{C}}}_{\:\boldsymbol{j}}(\mathbf{X})\cdot\boldsymbol{\mathsf{g}}_{\:\boldsymbol{j}}(\mathbf{a},\mathbf{b}).
\end{align}
Recall from section \ref{sec:intution} that $m_{\:t}(\mathbf{X},\mathbf{a},\mathbf{b})$ solves the following  for $m\in(a_{\:t},b_{\:t})$:
\begin{align*}
\bar{\boldsymbol{\mathsf{1}}}\big(\mathbf{X},\mathbf{a},\mathbf{b}_{\:\text{-}t}(m)\big)=\bar{\boldsymbol{\mathsf{1}}}\big(\mathbf{X},\mathbf{a}_{\:\text{-}t}(m),\mathbf{b})\big).
\end{align*}

We will see that for large $J$ and large $|\mathbf{X}|$, it can be shown that $\bar{\boldsymbol{\mathsf{1}}}_{\:J}(\mathbf{X},\mathbf{a},\mathbf{b})$ is a good approximation to $\bar{\boldsymbol{\mathsf{1}}}(\mathbf{X},\mathbf{a},\mathbf{b})$. Consequently, we define $\hat{m}_{\:t,\:J}(\mathbf{X},\mathbf{a},\mathbf{b})$ as a solution to the following equation for $m\in(a_{\:t},b_{\:t})$:
\begin{align}\label{kd:eq:b}
\bar{\boldsymbol{\mathsf{1}}}_{\:J}\big(\mathbf{X},\mathbf{a},\mathbf{b}^{-t}(m)\big)=\bar{\boldsymbol{\mathsf{1}}}_{\:J}\big(\mathbf{X},\mathbf{a}^{-t}(m),\mathbf{b})\big).
\end{align}

To obtain a parallel $k$-d tree construction, we approximate $m_{\:t}(\mathbf{X},\mathbf{a},\mathbf{b})$ with $\hat{m}_{\:t,\:J}(\mathbf{X},\mathbf{a},\mathbf{b})$ during the construction of the $k$-d tree, selecting the integer $J$ as a trade-off between approximation accuracy and speed. Since we have the relation $\bar{\boldsymbol{\mathsf{1}}}_{\:J}(\mathbf{X},\mathbf{a},\mathbf{b})=\sum_{\boldsymbol{j}\in\mathbb{N}_{\:J\:,p}}\mathbf{g}_{\:\boldsymbol{j}}(\mathbf{a},\mathbf{b})\cdot\bar{\boldsymbol{\mathsf{C}}}_{\:\boldsymbol{j}}(\mathbf{X})$, we realize that $\hat{m}_{\:t,\:J}(\mathbf{X},\mathbf{a},\mathbf{b})$ is a solution to an equation that is characterized by a set of statistics $\big\{\bar{\boldsymbol{\mathsf{C}}}_{\:\boldsymbol{j}}(\mathbf{X})\colon\boldsymbol{j}\in\mathbb{N}_{\:J\:,p}\big\}$ which can be computed entirely in parallel. Hence the computation of $\hat{m}_{\:t,\:J}(\mathbf{X},\mathbf{a},\mathbf{b})$ for $(\mathbf{a},\mathbf{b})\subseteq(\mathbf{0},\mathbf{1})$, and thus the construction of an approximate $k$-d tree, can be done straightforwardly within a parallel, distributed computing environment.

Define the error in approximation of $\mathsf{1}(\mathbf{a}<\mathbf{x}<\mathbf{b})$ at $\mathbf{x}$ as:
\begin{align*}
\mathsf{e}_{\:J}(\mathbf{x},\mathbf{a},\mathbf{b})=\mathsf{1}(\mathbf{a}<\mathbf{x}<\mathbf{b})-\mathsf{1}_{\:J}(\mathbf{x},\mathbf{a},\mathbf{b}).
\end{align*}
Also define the total error $\mathsf{E}_{\:J}(\mathbf{X},\mathbf{a},\mathbf{b})$ and the average error $\bar{\mathsf{E}}_{\:J}(\mathbf{X},\mathbf{a},\mathbf{b})$ as:
\begin{align*}
&\mathsf{E}_{\:J}(\mathbf{X},\mathbf{a},\mathbf{b})=\boldsymbol{\mathsf{1}}(\mathbf{X},\mathbf{a},\mathbf{b})-\boldsymbol{\mathsf{1}}_{\:J}(\mathbf{X},\mathbf{a},\mathbf{b}),\\ 
&\bar{\mathsf{E}}_{\:J}(\mathbf{X},\mathbf{a},\mathbf{b})=\bar{\boldsymbol{\mathsf{1}}}(\mathbf{X},\mathbf{a},\mathbf{b})-\bar{\boldsymbol{\mathsf{1}}}_{\:J}(\mathbf{X},\mathbf{a},\mathbf{b}).
\end{align*}
Finally, let us also define the set
\begin{align*}
\mathcal{N}_{\:p}=\{\mathcal{A}_{\:t}(m,\mathbf{a},\mathbf{b}):1\leq t\leq p,m\in(0,1),(\mathbf{a},\mathbf{b})\subseteq(\mathbf{0},\mathbf{1})\}.
\end{align*}

Now, we state two basic assumptions for the random variable $\tilde{\mathbf{x}}$.

\begin{assumption}\label{assumption:positive_density}
$P:\mathbb{B}(\mathbf{0},\mathbf{1})\mapsto[0,1]$ is a probability measure which is absolutely continuous with respect to the Lebesgue measure $\lambda_{\:p}$ on $\mathbb{B}(\mathbf{0},\mathbf{1})$.
\end{assumption}

\begin{assumption}\label{assumption:sup_null_set}
\begin{align*}
\sup_{\mathcal{A}_{\:t}(m,\mathbf{a},\mathbf{b})\in\mathcal{N}_{\:p}}\bigg(\lim_{|\mathbf{X}|\to\infty}\Big|\frac{1}{|\mathbf{X}|}\sum_{i\in I}\mathsf{1}\big(\mathbf{x}_{\:i}\in \mathcal{A}_{\:t}(m,\mathbf{a},\mathbf{b})\big)\Big|\bigg)\overset{a.s.}{\rightarrow}0.
\end{align*}
\end{assumption}

These assumptions enable us to obtain the following stochastic bounds on approximation accuracy.

\begin{theorem}\label{thm:n_l=n_r}
Suppose Assumption \ref{assumption:positive_density} holds. Then, given any $(\mathbf{a},\mathbf{b})\subseteq(\mathbf{0},\mathbf{1})$, we have:
\begin{align*}
&(A)\lim_{J\to\infty}E_{\:P}\big(\mathsf{e}_{\:J}(\tilde{\mathbf{x}},\mathbf{a},\mathbf{b})\big)=0. \\
&(B) \lim_{J\to\infty}\lim_{|\mathbf{X}|\to\infty}\bar{\mathsf{E}}(\mathbf{X},\mathbf{a},\mathbf{b})\overset{a.s.}{=}0. \\
&(C) \lim_{J\to\infty}\lim_{|\mathbf{X}|\to\infty}\bar{\boldsymbol{\mathsf{1}}}\Big(\mathbf{X},\mathbf{a},\mathbf{b}_{\:\text{-}t}\big(\hat{m}_{\:t,\:J}(\mathbf{X},\mathbf{a},\mathbf{b})\big)\Big) \\
&\hspace{4cm}=\lim_{J\to\infty}\lim_{|\mathbf{X}|\to\infty}\bar{\boldsymbol{\mathsf{1}}}\Big(\mathbf{X},\mathbf{a}_{\:\text{-}t}\big(\hat{m}_{\:t,\:J}(\mathbf{X},\mathbf{a},\mathbf{b})\big),\mathbf{b})\Big).
\end{align*}
\end{theorem}

\begin{proof}
See Appendix \ref{ap:stoch_proofs}.
\end{proof}

\begin{theorem}\label{thm:n_l=n/2}
Suppose Assumption \ref{assumption:positive_density} and Assumption \ref{assumption:sup_null_set} hold. Then, given any $(\mathbf{a},\mathbf{b})\subseteq(\mathbf{0},\mathbf{1})$, 
we have:
\begin{align*}
&\lim_{J\to\infty}\lim_{|\mathbf{X}|\to\infty}\bar{\boldsymbol{\mathsf{1}}}\Big(\mathbf{X},\mathbf{a},\mathbf{b}_{\:\text{-}t}\big(\hat{m}_{\:t,\:J}(\mathbf{X},\mathbf{a},\mathbf{b})\big)\Big)\overset{a.s.}{=}\frac{1}{2}\lim_{J\to\infty}\lim_{|\mathbf{X}|\to\infty}\bar{\boldsymbol{\mathsf{1}}}(\mathbf{X},\mathbf{a},\mathbf{b}), \\
&\lim_{J\to\infty}\lim_{|\mathbf{X}|\to\infty}\bar{\boldsymbol{\mathsf{1}}}\Big(\mathbf{X},\mathbf{a}_{\:\text{-}t}\big(\hat{m}_{\:t,\:J}(\mathbf{X},\mathbf{a},\mathbf{b}),\mathbf{b}\big)\Big)\overset{a.s.}{=}\frac{1}{2}\lim_{J\to\infty}\lim_{|\mathbf{X}|\to\infty}\bar{\boldsymbol{\mathsf{1}}}(\mathbf{X},\mathbf{a},\mathbf{b}).
\end{align*}
\end{theorem}

\begin{proof}
See Appendix \ref{ap:stoch_proofs}.
\end{proof}

\section{Constructing a balanced $k$-d tree from distributed data}\label{sec:algorithm}

Let us briefly recall the canonical construction of a $k$-d tree and its equivalent formulation in terms of conditional medians $m_{\:t}(\mathbf{X},\mathbf{a},\mathbf{b})$ as discussed in section \ref{sec:intution}. To build a $k$-d tree of depth $D$ from data $\mathbf{X}$, we compute $2^D-1$ conditional medians $m_{\:1,\:1}(\mathbf{X})$; $m_{\:2,\:1}(\mathbf{X})$, $m_{\:2,\:2}(\mathbf{X})$; $\dots$; $m_{\:D,\:1}(\mathbf{X})$, $\dots$, $m_{\:D,\:2^{D-1}}(\mathbf{X})$, and these $2^D-1$ statistics provide us necessary and sufficient information to build the entire $k$-d tree. This is a naturally recursive procedure, in that $m_{\:2,\:1}(\mathbf{X})$ and $m_{\:2,\:2}(\mathbf{X})$ are dependent on $m_{\:1,\:1}(\mathbf{X})$, and so forth.

By contrast, given a parameter $J$, our approximate $k$-d tree construction will instead proceed by non-recursively computing $2^D-1$ approximate conditional medians $\hat{m}_{\:1,\:1,\:J}(\mathbf{X})$; $\hat{m}_{\:2,\:1,\:J}(\mathbf{X})$, $\hat{m}_{\:2,\:2,\:J}(\mathbf{X})$; $\dots$; $\hat{m}_{\:D,\:1,\:J}(\mathbf{X})$, $\dots$, $\hat{m}_{\:D,\:2^{D-1},\:J}(\mathbf{X})$. (To be clear, we have: $\hat{m}_{\:1,\:1,\:J}(\mathbf{X})=\hat{m}_{\:1,\:J}(\mathbf{X},\mathbf{0},\mathbf{1})$, $\hat{m}_{\:2,\:1,\:J}(\mathbf{X})=\hat{m}_{\:2,\:J}\big(\mathbf{X},\mathbf{0},\mathbf{1}_{\:\text{-}1}\big(\hat{m}_{\:1,\:J}(\mathbf{X},\mathbf{0},\mathbf{1})\big)\big)$, etc.) 

Observe that $\hat{m}_{\:d,\:k}^J(\mathbf{X})$ is formally a function of the statistics $\{\bar{\boldsymbol{\mathsf{C}}}_{\:\boldsymbol{j}}(\mathbf{X}):\boldsymbol{j}\in\mathbb{N}_{\:J\:,p}$ for $1\leq d\leq D,1\leq k\leq 2^{d-1}$ which are statistics that can be computed entirely in parallel. Thus we have an approximate $k$-d tree construction in which the input data can be read once to compute $\bar{\boldsymbol{\mathsf{C}}}_{\:\boldsymbol{j}}(\mathbf{X})$ for each $\boldsymbol{j}\in\mathbb{N}_{\:J\:,p}$. Then, equations \ref{kd:eq:a} and \ref{kd:eq:b} are used recursively to compute $\hat{m}_{\:d,\:k}^J(\mathbf{X})$ for $1\leq d\leq D,1\leq k\leq 2^{d-1}$.

It is important to consider the metric by which approximation accuracy should be judged. For data analysis, a main purpose of the $k$-d tree construction is to divide a data set into neighborhoods containing almost equal numbers of observations. In the canonical construction of a $k$-d tree, if we start with $|\mathbf{X}|$ observations, at depth $D$, each of the cell neighborhoods contains approximately $\big(\nicefrac{1}{2^D}\big)\cdot|\mathbf{X}|$ observations (assuming we choose to discard any point on the cell boundaries). When we are approximating $m_{\:d,\:k}^J(\mathbf{X})$ by $\hat{m}_{\:d,\:k,\:J}(\mathbf{X})$, even if we are off by a small amount, as long as the eqi-cardinality of the neighborhoods is maintained, we have succeeded in constructing a sufficiently accurate $k$-d tree to serve the purpose at hand. 

Thus, instead of judging how well the approximate median $\hat{m}_{\:d,\:k,\:J}(\mathbf{X})$ approximates the exact median $m_{\:d,\:k}(\mathbf{X})$, for $1\leq d\leq D,1\leq k\leq 2^{d-1}$, we focus on how well the approximate $k$-d tree subdivides the overall data point cloud into cell neighborhoods having almost equal numbers of observations. At depth $D$ we have $2^D$ neighborhoods and so to judge the accuracy of the $k$-d tree, we can look at the range of all cell boundary counts and determine their distances from the ideal cell boundary counts of $\big(\nicefrac{1}{2^D}\big)\cdot|\mathbf{X}|$ observations for each cell.

Suppose that during the construction of a $k$-d tree, a general neighborhood $(\mathbf{a},\mathbf{b})$ containing $n$ observations gets split into disjoint sub-neighborhoods $(\mathbf{a}_{\:\text{left}},\mathbf{b}_{\:\text{left}})$ and $(\mathbf{a}_{\:\text{right}},\mathbf{b}_{\:\text{right}})$ containing $\hat{n}_{\:l}$ and $\hat{n}_{\:r}$ observations. For future reference, we will be calling the neighborhood $(\mathbf{a},\mathbf{b})$ as the parent neighborhood of its two children neighborhoods $(\mathbf{a}_{\:\text{left}},\mathbf{b}_{\:\text{left}})$ and $(\mathbf{a}_{\:\text{right}},\mathbf{b}_{\:\text{right}})$. Now from Theorem \ref{thm:n_l=n/2}, if Assumption \ref{assumption:positive_density} and  Assumption \ref{assumption:sup_null_set} holds, then given a $\epsilon>0$, we can say that, for large enough $|\mathbf{X}|$ and $J$, we will have:
\begin{align*}
&\big|\bar{\boldsymbol{\mathsf{1}}}(\mathbf{X},\mathbf{a}_{\:\text{left}},\mathbf{b}_{\:\text{left}})-\frac{1}{2}\bar{\boldsymbol{\mathsf{1}}}(\mathbf{X},\mathbf{a},\mathbf{b})\big|<\epsilon,\text{ and } \\
&\big|\bar{\boldsymbol{\mathsf{1}}}(\mathbf{X},\mathbf{a}_{\:\text{right}},\mathbf{b}_{\:\text{right}})-\frac{1}{2}\bar{\boldsymbol{\mathsf{1}}}(\mathbf{X},\mathbf{a},\mathbf{b})\big|<\epsilon.
\end{align*}
In other words, for an arbitrary interval with $n$ elements, we will have:
\begin{align*}
\big|\hat{n}_{\:l}-\frac{n}{2}\big|<\epsilon\:|\mathbf{X}|\text{ and }\big|\hat{n}_{\:r}-\frac{n}{2}\big|<\epsilon\:|\mathbf{X}|.
\end{align*}

Let $\hat{n}_{\:d,\:k,\:J}(\mathbf{X})$ denote the cell count of the $k$th cell at depth $d$, for parameter $J$, for $1\leq d\leq D,1\leq k\leq 2^d$. For consistency, we let $\hat{n}_{\:0,\:1,\:J}(\mathbf{X})=|\mathbf{X}|$, and we assume that at level $0$, there is only one cell containing all the observations, and in general, we have $2^d$ cells at level $d$. At depth $D$, consider the $k$th cell, and let $k_{\:(0)},k_{\:(1)},\dots,k_{\:(D)}$ be the sequence of parent cell indices of this cell, so that $k_{\:(0)}=1$, $k_{\:(D)}=k$ and  $1\leq k_{(d)}\leq 2^d$ for $1\leq d<D$. Then for large enough $|\mathbf{X}|$ and $J$, we can say from last paragraph that:
\begin{align*}
\big|\:\hat{n}_{\:d,\:k_{\:(d)},\:J}(\mathbf{X})-\frac{1}{2}\hat{n}_{\:d-1,\:k_{\:(d-1)},\:J}(\mathbf{X})\:\Big|<\epsilon\:|\mathbf{X}|\text{ for }1\leq d\leq D.
\end{align*}
So, for large enough $|\mathbf{X}|$ and $J$, we have:
\begin{align*}
&\hspace{-2em}\Big|\:\hat{n}_{\:D,\:k,\:J}(\mathbf{X})-\frac{1}{2^D}\hat{n}_{\:0,\:1,\:J}(\mathbf{X})\:\Big| \\
&\leq \sum_{d=1}^D\Big|\:\frac{1}{2^{D-d}}\hat{n}_{\:d,\:k_{\:(d)},\:J}(\mathbf{X})-\frac{1}{2^{D-d+1}}\hat{n}_{\:d-1,\:k_{\:(d-1)},\:J}(\mathbf{X})\:\Big| \\
&=\sum_{d=1}^D\frac{1}{2^{D-d}}\Big|\:\hat{n}_{\:d,\:k_{\:(d)},\:J}(\mathbf{X})-\frac{1}{2}\hat{n}_{\:d-1,\:k_{\:(d-1)},\:J}(\mathbf{X})\:\Big| \\
&<\sum_{d=1}^D\frac{\epsilon\:|\mathbf{X}|}{2^{d-1}}<2\epsilon\:|\mathbf{X}|.
\end{align*}
In other words, if we assume the technical conditions of Assumption \ref{assumption:positive_density} and Assumption \ref{assumption:sup_null_set}, we indeed have the desired result that
\begin{align*}
\Big|\hat{n}_{\:D,\:k,\:J}(\mathbf{X})-\frac{1}{2^D}|\mathbf{X}|\Big|< 2\epsilon\:|\mathbf{X}|\text{ for }1\leq k\leq 2^D
\end{align*}
as both $|\mathbf{X}|$ and $J$ become large. Thus our approximate $k$-d tree construction is assured to become asymptotically accurate as the data set size increases and we consider an increasingly large approximation order parameter $J$.

\section{Accuracy and computational scalability}\label{sec:performance}

We now study the accuracy and scalablity of this method of $k$-d tree construction algorithm using simulated data. An implementation of this construction in map-reduce is provided in Appendix \ref{ap:construct}.

We simulated observations $\tilde{\mathbf{x}}=(\tilde{x},\tilde{y},\tilde{z})~$ from a multivariate normal distribution with mean $\mathbf{0}_{\:3}'$ and common correlation $\rho$ between each distinct pair of variables. We considered values for $\rho$ in the set $\{0,0.25,0.5,0.75\}$, and simulated $N=3\times 10^9$ observations for each value.

\subsection{Assessment of accuracy}\label{sec:accuracy}

We first constructed a $k$-d tree for each $\rho$, and then generated box plots to observe distributions of cell counts at each level of $k$-d tree depth from $6$ to $10$. The exact $k$-d tree construction is impractical at this scale, and so we judge accuracy in terms of cell counts rather than cell boundaries.

\begin{figure}[t]
\begin{center}
\includegraphics[trim={.3in 0 0 0}, clip, width=3.5in, angle=270]{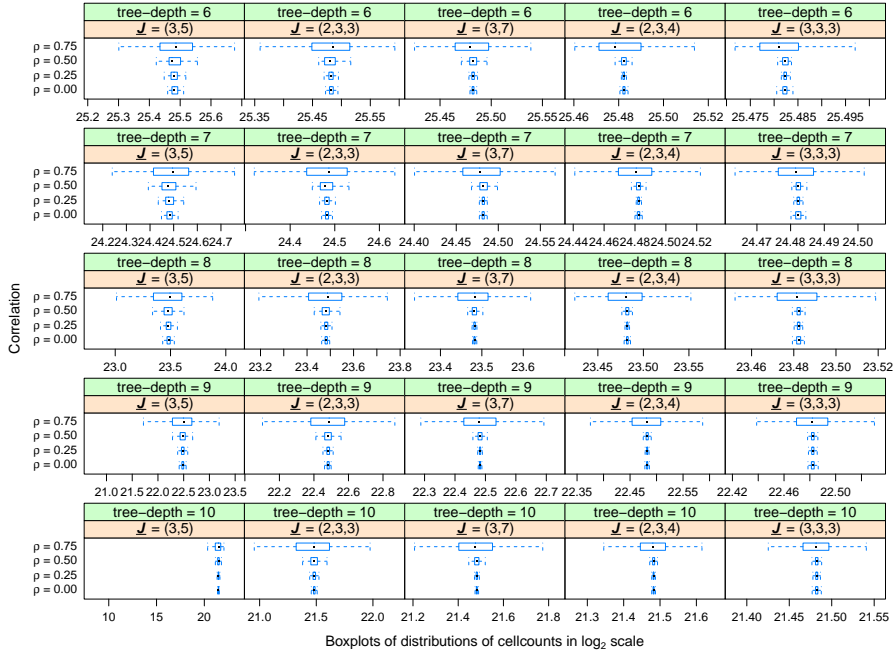}
\end{center}
\caption{\label{fig:norm_acc}Balanced tree accuracy measured by cell counts, for simulated normal data in three dimensions.}
\end{figure}

Figure \ref{fig:norm_acc} shows the results of this procedure for data simulated from the normal distribution described above, where we take $\tilde{x}$, $\tilde{y}$ and $\tilde{z}$ as our variables to construct the approximate $k$-d tree. Each panel represents a collection of $k$-d trees corresponding to depth $D$ and accuracy parameter $\underline{\boldsymbol{J}}$ (described in section \ref{sec:eff}), as different rows represent different depth and different columns represent different accuracy parameters. Inside each panel we have four different box plots corresponding to four levels of correlation $\rho$. 

For each combination of depth $D$,  accuracy parameter  $\underline{\boldsymbol{J}}$, and correlation $\rho$, the corresponding box plot demonstrates the distribution of $k$-d tree neighborhood cell counts. For example, the third (from bottom) box plot in the top-left panel demonstrates neighborhood cell counts of $k$-d tree output for depth $D=6$,  accuracy parameter $\underline{\boldsymbol{J}}=(3,5)$ and $\rho=0.5$. At depth $D=6$ we should have $2^6=64$ neighborhoods, for exact canonical construction we should have $(3\times 10^9)/64\simeq 2^{25.48232}$ observations in each neighborhood. Here, the box plot shows these neighborhood cell counts range from about $2^{25.45}$ to $2^{25.52}$, and the box labels shows half of the cell counts range from $2^{25.46}$ to $2^{25.49}$.

As a rule of thumb we can say that a narrower box plot corresponds to a more accurate $k$-d tree construction by this measure. We may also make the following general qualitative observations regarding Figure \ref{fig:norm_acc}. First, the modulus of the vector coefficient $\underline{\boldsymbol{J}}$ increases, the accuracy of the $k$-d tree output indeed generally increases. Second, the tree level increases, the accuracy of the $k$-d tree output generally decreases. This is expected because if we approximate the cell boundaries badly at lower levels, we will approximate the cell boundaries much worse at higher levels, since they are dependent on cell boundaries badly at lower levels. Finally, as correlation $\rho$ increases, the accuracy decreases in general. (However, this is harder to identify at higher levels.)

\subsection{Computational scalability}\label{sec:runtime}

\begin{figure}[t]
\begin{center}
\includegraphics[trim={.3in .3in 0 0}, clip, width=3.5in, angle=270]{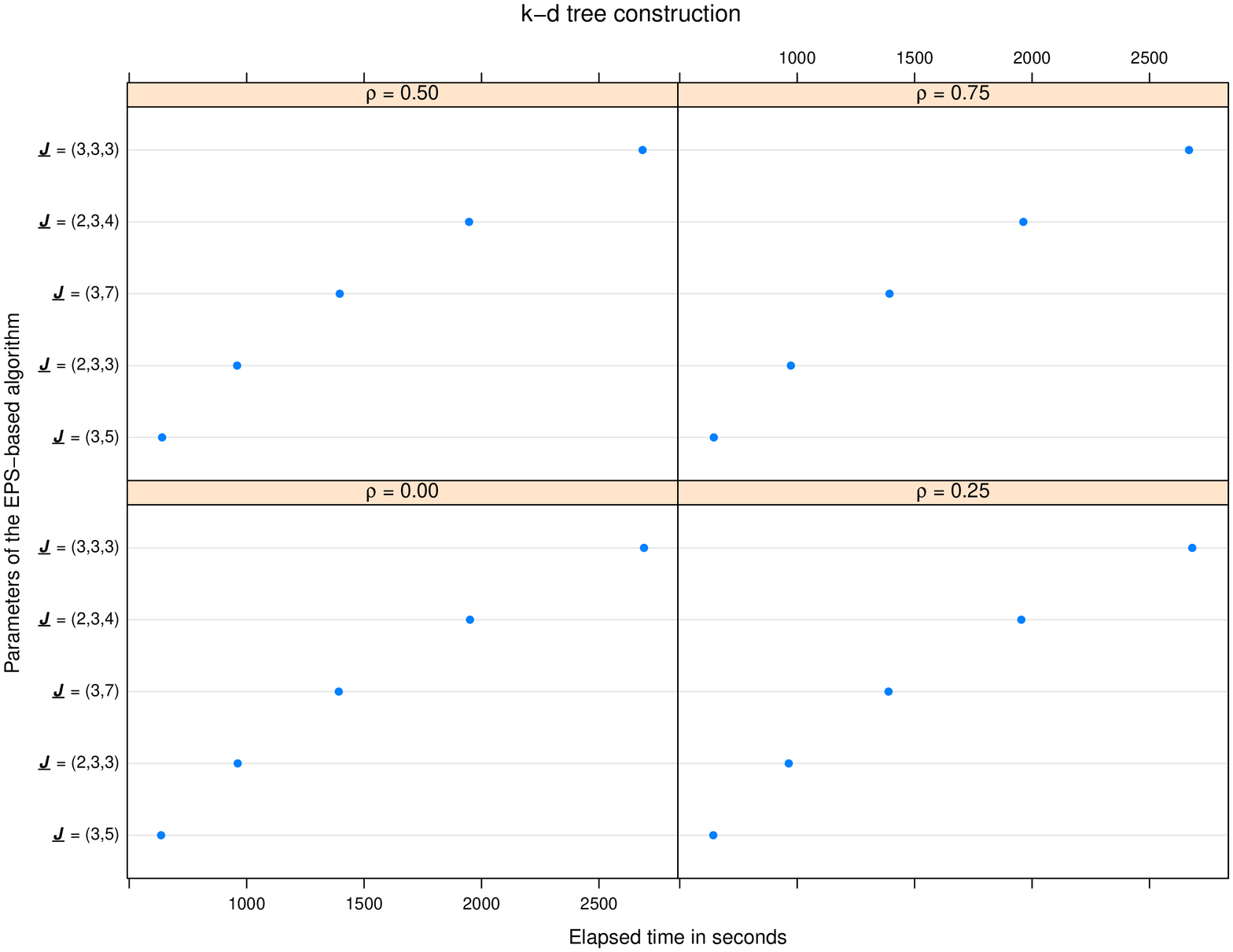}
\end{center}
\caption{\label{fig:norm_run}Running times for approximate $k$-d tree construction with increasing modulus of accuracy parameter $\underline{\boldsymbol{J}}$.}
\end{figure}

Figure \ref{fig:norm_run} quantifies the running time of this approximate $k$-d tree construction for different accuracy parameters $\underline{\boldsymbol{J}}$ and correlations $\rho$. Observe that the running time does not depend heavily on depth $D$, because the map-reduce step is the major driver, and post-reduce local optimizations here take negligible time compared to the main map-reduce step. In general, running time increases as $p$ increases and the modulus of the accuracy parameter $\underline{\boldsymbol{J}}$ increases, but that $\rho$ does not appear to have any major influence.

\section{Discussion}\label{sec:disc}

In this article we have described a scalable mechanism to construct $k$-d trees for large data sets, which can be implemented naturally within a distributed computing environment using programming models such as map-reduce. While the median computations necessary for exact construction of a canonical $k$-d tree cannot be done in parallel, we have instead proposed a parallel median approximation algorithm which enables all necessary quantities to construct a canonical $k$-d tree to be computed using only a single pass over the input data. We have shown this approach to come with a variety of theoretical guarantees as described in section \ref{sec:main}, focusing in particular on its ability to produce balanced $k$-d trees as discussed in section \ref{sec:algorithm}. In turn, section \ref{sec:performance} has provided a simulation study quantifying both the accuracy and speed of a concrete map-reduce implementation using artificially generated data, as a function of various algorithm parameters and data properties.

\appendix

\section{Proofs of results from section \ref{subsec:approx}}\label{ap:approx_proofs}

First, we consider a Fourier expansion of the indicator function $\mathsf{1}(z<0)$. We have, for $-\pi<z<\pi; z\neq 0$:
\begin{align*}
\mathsf{1}(z<0)=\mathsf{1}(z\leq 0)=\frac{1}{2}-\frac{2}{\pi}\sum_{j=1}^{\infty}\frac{\sin\big((2j-1)z\big)}{(2j-1)}
\end{align*}
Now if $x\in(0,1)$ and $c$ is a constant in $(0,1)$, we have $-\pi<-2<x-c<2<\pi$, or $x-c\in(-\pi,\pi)$. So we can expand $\mathsf{1}(x<c)=\mathsf{1}(x-c<0)$ and $\mathsf{1}(x\leq c)=\mathsf{1}(x-c\leq 0)$ as follows:
\begin{align}\label{proof:eq:1}
\mathsf{1}(x<c)&=\mathsf{1}(x\leq c)=\frac{1}{2}-\sum_{j=1}^{\infty}\frac{2}{\pi(2j-1)}\sin\big((2j-1)(x-c)\big)\\
&=\frac{1}{2}+\sum_{j=1}^{\infty}\bigg(\frac{2\sin\big((2j-1)c\big)}{\pi(2j-1)}\cdot\cos\big((2j-1)x\big) \nonumber \\
&\hspace{5cm}+\frac{-2\cos\big((2j-1)c\big)}{\pi(2j-1)}\cdot\sin\big((2j-1)x\big)\bigg). \nonumber
\end{align}

The above expression is recognizable as a convergent $L^2$ approximation from the class of sum of separable trigonometric functions of the form $\sum_{j}f_{\:j}(c)\cdot g_{\:j}(x)$ \cite[Example 3.10]{Chakravorty2021}.

\subsection{Proof of Lemma \ref{lemma:1d:exp}}
\begin{proof}
(A)\ Let $x\in(0,1)$ and consider arbitrary constants $a,b$ such that $0\leq a<x<b\leq 1$. Observe that if $a<b$, then $\mathsf{1}(a<x<b)=\mathsf{1}(a<x)-\mathsf{1}(b\leq x)$. The idea is to approximate these terms along the lines of equation \ref{proof:eq:1}. However, if $a=0$, then we already know that as $x\in(0,1)$, so, we must have $a<x$, and in that case $\mathsf{1}(a<x)=1$. Similarly if $b=1$, then $\mathsf{1}(b\leq x)=0$. Now we have:
\begin{align}\label{proof:eq:2}
&\mathsf{1}(a<x<b)=\mathsf{1}(a=0)\cdot\mathsf{1}(b=1)\cdot 1+\mathsf{1}(a>0)\cdot\mathsf{1}(b=1)\cdot\mathsf{1}(a<x) \\
&+\mathsf{1}(a=0)\cdot\mathsf{1}(b<1)\cdot\mathsf{1}(x< b)+\mathsf{1}(a>0)\cdot\mathsf{1}(b<1)\cdot\mathsf{1}(a<x<b)\nonumber \\
&=\big(1-\mathsf{1}(a>0)\big)\cdot\big(1-\mathsf{1}(b<1)\big)+\mathsf{1}(a>0)\cdot\big(1-\mathsf{1}(b<1)\big)\cdot\mathsf{1}(a<x)\nonumber \\
&+\big(1-\mathsf{1}(a>0)\big)\cdot\mathsf{1}(b<1)\cdot\big(1-\mathsf{1}(b\leq x)\big)\nonumber \\
&+\mathsf{1}(a>0)\cdot\mathsf{1}(b<1)\cdot\big(\mathsf{1}(a<x)-\mathsf{1}(b\leq x)\big)\nonumber \\
&=1-\mathsf{1}(a>0)+\mathsf{1}(a>0)\cdot\mathsf{1}(a<x)-\mathsf{1}(b<1)\cdot\mathsf{1}(b\leq x).\nonumber
\end{align}
If we replace $\mathsf{1}(a<x)$ and $\mathsf{1}(b\leq x)$ in equation \ref{proof:eq:2}, with the expansion in equation \ref{proof:eq:1}, we have:
\begin{align*}
&\mathsf{1}(a<x<b)=\bigg(1-\frac{\mathsf{1}(a>0)}{2}-\frac{\mathsf{1}(b<1)}{2}\bigg) \\
&+\sum_{j=1}^{\infty}\frac{2.\mathsf{1}(b<1)\sin\big((2j-1)b\big)-2.\mathsf{1}(a>0)\sin\big((2j-1)a\big)}{\pi(2j-1)}\cos\big((2j-1)x\big) \\
&+\sum_{j=1}^{\infty}\frac{2.\mathsf{1}(a>0)\cos\big((2j-1)a\big)-2.\mathsf{1}(b<1)\cos\big((2j-1)b\big)}{\pi(2j-1)}\sin\big((2j-1)x\big).
\end{align*}
Let us define, for $x\in(0,1)$ and constants $a\in[0,1),b\in(0,1]$, the functions:
\begin{align*}
&\mathsf{c}_{\:0}(x)=1,\quad \mathsf{c}_{\:2j-1}(x)=\cos\big((2j-1)x\big),\quad \mathsf{c}_{\:2j}(x)=\sin\big((2j-1)x\big);\text{ and } \\
&\mathsf{g}_{\:0}(a,b)=1-\frac{1}{2}\big(\mathsf{1}(a>0)+\mathsf{1}(b<1)\big), \\
&\mathsf{g}_{\:2j-1}(a,b)=\frac{2}{\pi(2j-1)}\Big(\mathsf{1}(b<1)\sin\big((2j-1)b\big)-\mathsf{1}(a>0)\sin\big((2j-1)a\big)\Big), \\ 
&\mathsf{g}_{\:2j}(a,b)=\frac{2}{\pi(2j-1)}\Big(\mathsf{1}(a>0)\cos\big((2j-1)a\big)-\mathsf{1}(b<1)\cos\big((2j-1)b\big)\Big).
\end{align*}
Then we have for $x\in(0,1)$, $a\in[0,1)$, $b\in(0,1]$ and $x\neq a$, $x\neq b$
\begin{align*}
\mathsf{1}(a<x<b)=\sum_{j=0}^{\infty}\mathsf{c}_{\:j}(x)\cdot\mathsf{g}_{\:j}(a,b).
\end{align*}
\end{proof}

\subsection{Proof of Lemma \ref{lemma:1d:unif}}

\begin{proof}
(A)\ From \citet[part (A) of Lemma E.3]{Chakravorty2021}, we realize that if $0<a<1$, then $\mathsf{1}_{\:J}(x-a)$ uniformly converges to $\mathsf{1}(0<x-a)$ if $x-a\in(-\pi,-\delta)\bigcup(\delta,\pi)$ or $x\in(a-\pi,a-\delta)\bigcup(a+\delta,a+\pi)$. Observe that $a-\pi<0$ and $1<a+\pi$ for $0<a<1$. If $a=0$ then $\mathsf{1}_{\:J}(x-a)=\mathsf{1}(x<a)=0$. So, $\mathsf{1}_{\:J}(x-a)$ uniformly converges to $\mathsf{1}(0<x-a)$, if $x\in U_{\delta,(a,b)}$. Similarly, we can show $\mathsf{1}_{\:J}(x-b)$ uniformly converges to $\mathsf{1}(0\leq x-b)$, if $x\in U_{\delta,(a,b)}$. Then, $\mathsf{1}^{J}(x,a,b)$ converges uniformly to its limit $\mathsf{1}(a<x<b)$ in $U_{\delta,(a,b)}$.

(B)\ Since $\mathsf{1}_{\:J}(x,a,b)=1-\mathsf{1}(a>0)+\mathsf{1}(a>0)\mathsf{1}_{\:J}(x-a)-\mathsf{1}(b<1)\mathsf{1}_{\:J}(x-b)$, we conclude from \citet[part (B) of Lemma E.3]{Chakravorty2021} that the sequence of functions $\mathsf{1}_{\:J}(x,a,b)$ is uniformly bounded in the interval $(0,1)$.
\end{proof}

\subsection{Proof of Lemma \ref{lemma:pd:unif}}
\begin{proof}
For $\mathbf{a}\in[\mathbf{0},\mathbf{1}),\mathbf{b}\in(\mathbf{0},\mathbf{1}]$, define
\begin{align*}
P_{\:(\mathbf{a},\mathbf{b})}=\times_{l=1}^pP_{\:(a_{\:l},b_{\:l})}, \quad
U_{\:\delta,\:(\mathbf{a},\mathbf{b})}=\times_{l=1}^pU_{\:\delta,\:(a_{\:l},b_{\:l})}.
\end{align*}
Now fix $\delta>0$, and for $1\leq l\leq p$, let us define
\begin{align*}
f_{\:J,\:l}(\mathbf{x},\mathbf{a},\mathbf{b})=\mathsf{1}_{\:J}(x_{\:l},a_{\:l},b_{\:l}), \quad
f_{\:l}(\mathbf{x},\mathbf{a},\mathbf{b})=\mathsf{1}(a_{\:l}<x_{\:l}<b_{\:l}).
\end{align*}
Thus we have $\mathsf{1}_{\:J}(\mathbf{x},\mathbf{a},\mathbf{b})=\prod_{l=1}^pf_{\:J,\:l}(\mathbf{x},\mathbf{a},\mathbf{b})$ and $\mathsf{1}(\mathbf{x},\mathbf{a},\mathbf{b})=\prod_{l=1}^pf_{\:l}(\mathbf{x},\mathbf{a},\mathbf{b})$.

Now, from parts (A) and (B) of Lemma \ref{lemma:1d:exp}, we know that $\mathsf{1}_{\:J}(x_{\:l},a_{\:l},b_{\:l})$ converges uniformly to $\mathsf{1}(a_{\:l}<x_{\:l}<b_{\:l})$ when $x_{\:l}\in U_{\delta,(a_{\:l},b_{\:l})}$, and furthermore is uniformly bounded for $1\leq l\leq p$. So, $f_{\:J,\:l}(\mathbf{x},\mathbf{a},\mathbf{b})$ converges uniformly to $f_{\:l}(\mathbf{x},\mathbf{a},\mathbf{b})$ when $\mathbf{x}\in U_{\:\delta,\:(\mathbf{a},\mathbf{b})}$, and is uniformly bounded for $1\leq l\leq p$. Now, we apply the result of \citet[Theorem 7.9]{Rudin1976} for a product of $p$ functions to conclude that $\mathsf{1}_{\:J}(\mathbf{x},\mathbf{a},\mathbf{b})$ converges uniformly to $\mathsf{1}(\mathbf{a}<\mathbf{x}<\mathbf{b})$ for $\mathbf{x}\in U_{\:\delta,\:(\mathbf{a},\mathbf{b})}$.
\end{proof}

\section{Proofs of results from section \ref{subsec:stoch}}\label{ap:stoch_proofs}

\subsection{Proof of Theorem \ref{thm:n_l=n_r}}\label{pf:n_l=n_r}

\begin{proof}
(A) Let $\epsilon>0$. From part (B) of Lemma \ref{lemma:pd:unif}, we know that $\mathsf{1}_{\:J}(\mathbf{x},\mathbf{a},\mathbf{b})$ is uniformly bounded on $(\mathbf{0},\mathbf{1})$. Since the indicator function $\mathsf{1}(\mathbf{x},\mathbf{a},\mathbf{b})$ is also bounded, the absolute difference $|\mathsf{e}_{\:J}(\mathbf{x},\mathbf{a},\mathbf{b})|$ is hence uniformly bounded in $(\mathbf{0},\mathbf{1})$. Let this bound be denoted by $M$, so that $\big|\:\mathsf{e}_{\:J}(\mathbf{x},\mathbf{a},\mathbf{b})\:\big|<M$, for any $J$ and any $\mathbf{x}\in(\mathbf{0},\mathbf{1})$. 

Now, since $P$ is absolutely continuous with respect to $\lambda$, there exists an $\eta>0$ such that if $\lambda(A)<\eta$, then $P(A)<\nicefrac{\epsilon}{(2\cdot M)}$ for any $A\in\mathbb{B}(\mathbf{0},\mathbf{1})$. Let us hence pick a $\delta$ such that $\lambda\big({Q_{\:\delta,\:(\mathbf{a},\mathbf{b})}}'\big)=(2\cdot\delta)^p<\eta$, so that we have $P\big({Q_{\:\delta,\:(\mathbf{a},\mathbf{b})}}'\big)<\nicefrac{\epsilon}{2M}$. Next, because of uniform convergence on $Q_{\:\delta,\:(\mathbf{a},\mathbf{b})}$, by part (A) of Lemma \ref{lemma:pd:unif} we can choose a $J$ large enough to have $\big|\:\mathsf{e}_{\:J}(\mathbf{x},\mathbf{a},\mathbf{b})\:\big|=\big|\:\mathsf{1}(\mathbf{x},\mathbf{a},\mathbf{b})-\mathsf{1}_{\:J}(\mathbf{x},\mathbf{a},\mathbf{b})\:\big|<\nicefrac{\epsilon}{2}$, for any $\mathbf{x}\in Q_{\:\delta,\:(\mathbf{a},\mathbf{b})}$. Then we have:
\begin{align*}
&\hspace{-1em}\big|\:E_{\:P}\big(\mathsf{e}_{\:J}(\tilde{\mathbf{x}},\mathbf{a},\mathbf{b})\big)\:\big|=\big|\int\limits_{(\mathbf{0},\mathbf{1})}\mathsf{e}_{\:J}(\mathbf{x},\mathbf{a},\mathbf{b})\:\mathrm{d}P\:\big| \\
&=\big|\int\limits_{Q_{\:\delta,\:(\mathbf{a},\mathbf{b})}}\mathsf{e}_{\:J}(\mathbf{x},\mathbf{a},\mathbf{b})\:\mathrm{d}P+\int\limits_{{Q_{\:\delta,\:(\mathbf{a},\mathbf{b})}}'}\mathsf{e}_{\:J}(\mathbf{x},\mathbf{a},\mathbf{b}) \:\mathrm{d}P\:\big| \\
&\leq \big|\int\limits_{Q_{\:\delta,\:(\mathbf{a},\mathbf{b})}}\mathsf{e}_{\:J}(\mathbf{x},\mathbf{a},\mathbf{b})\:\mathrm{d}P\:\big|+\big|\int\limits_{{Q_{\:\delta,\:(\mathbf{a},\mathbf{b})}}'}\mathsf{e}_{\:J}(\mathbf{x},\mathbf{a},\mathbf{b}) \:\mathrm{d}P\:\big| \\
&\leq\int\limits_{Q_{\:\delta,\:(\mathbf{a},\mathbf{b})}}|\:\mathsf{e}_{\:J}(\mathbf{x},\mathbf{a},\mathbf{b})\:|\:\mathrm{d}P+\int\limits_{{Q_{\:\delta,\:(\mathbf{a},\mathbf{b})}}'}|\:\mathsf{e}_{\:J}(\mathbf{x},\mathbf{a},\mathbf{b})\:|\:\mathrm{d}P \\
&<\int\limits_{Q_{\:\delta,\:(\mathbf{a},\mathbf{b})}}\nicefrac{\epsilon}{2}\cdot\mathrm{d}P+\int\limits_{{Q_{\:\delta,\:(\mathbf{a},\mathbf{b})}}'}M\cdot\mathrm{d}P=\nicefrac{\epsilon}{2}\cdot P\big(Q_{\:\delta,\:(\mathbf{a},\mathbf{b})}\big)
+M\cdot P\big({Q_{\:\delta,\:(\mathbf{a},\mathbf{b})}}'\big) \\
&\leq\nicefrac{\epsilon}{2}\cdot 1+M\cdot P\big({Q_{\:\delta,\:(\mathbf{a},\mathbf{b})}}'\big)<\nicefrac{\epsilon}{2}+M\cdot\nicefrac{\epsilon}{2M}<\epsilon.
\end{align*}

Consequently, we conclude that $E_{\:P}(\mathsf{e}_{\:J}\big(\tilde{\mathbf{x}},\mathbf{a},\mathbf{b})\big)\to 0$ as $J\to\infty$.

(B) Recall that the set of all $\mathbf{x}_{\:i}$ for $i\in I$ comprise independent and identically distributed observations of the random variable $\tilde{\mathbf{x}}$. Since the sequence $\Big\{E_{\:P}(\mathsf{e}_{\:J}\big(\tilde{\mathbf{x}},\mathbf{a},\mathbf{b})\big)\Big\}_{J=0}^{\infty}$ converges to a limit, it is uniformly bounded. Then, an application of Kolmogorov's strong law of large numbers gives us: 
\begin{align*}
\lim_{|X|\to\infty}\bar{\mathsf{E}}_{\:J}(\mathbf{X},\mathbf{a},\mathbf{b})\overset{a.s.}{=}E_{\:P}\big(\mathsf{e}_{\:J}(\tilde{\mathbf{x}},\mathbf{a},\mathbf{b})\big)\:[P].
\end{align*}
Then, from part (A), we have,
\begin{align*}
\lim_{J\to\infty}\big(\lim_{|X|\to\infty}\bar{\mathsf{E}}_{\:J}(\mathbf{X},\mathbf{a},\mathbf{b})\big)\overset{a.s.}{=}0\:[P].
\end{align*}

(C) From part (B), we have:
\begin{align}\label{main:proof:eq:a}
&\hspace{-1em}\lim_{J\to\infty}\lim_{|\mathbf{X}|\to\infty}\bar{\mathsf{E}}_{\:J}\Big(\mathbf{X},\mathbf{a},\mathbf{b}_{\:\text{-}t}\big(\hat{m}_{\:t,\:J}(\mathbf{X},\mathbf{a},\mathbf{b})\big)\Big)\overset{a.s.}{=}0\:[P]\nonumber \\
&\leftrightarrow\lim_{J\to\infty}\lim_{|\mathbf{X}|\to\infty}\bigg(\bar{\boldsymbol{\mathsf{1}}}_{\:J}\Big(\mathbf{X},\mathbf{a},\mathbf{b}_{\:\text{-}t}\big(\hat{m}_{\:t,\:J}(\mathbf{X},\mathbf{a},\mathbf{b})\big)\Big)\nonumber \\
&\hspace{4cm}-\bar{\boldsymbol{\mathsf{1}}}\Big(\mathbf{X},\mathbf{a},\mathbf{b}_{\:\text{-}t}\big(\hat{m}_{\:t,\:J}(\mathbf{X},\mathbf{a},\mathbf{b})\big)\Big)\bigg)\overset{a.s.}{=}0\:[P]\nonumber \\
&\leftrightarrow\lim_{J\to\infty}\lim_{|\mathbf{X}|\to\infty}\bar{\boldsymbol{\mathsf{1}}}_{\:J}\Big(\mathbf{X},\mathbf{a},\mathbf{b}_{\:\text{-}t}\big(\hat{m}_{\:t,\:J}(\mathbf{X},\mathbf{a},\mathbf{b})\big)\Big) \\
&\hspace{4cm}\overset{a.s.}{=}\lim_{J\to\infty}\lim_{|\mathbf{X}|\to\infty}\bar{\boldsymbol{\mathsf{1}}}\Big(\mathbf{X},\mathbf{a},\mathbf{b}_{\:\text{-}t}\big(\hat{m}_{\:t,\:J}(\mathbf{X},\mathbf{a},\mathbf{b})\big)\Big)\:[P]\nonumber.
\end{align}

Similarly, we also have:
\begin{align}\label{main:proof:eq:b}
&\lim_{J\to\infty}\lim_{|\mathbf{X}|\to\infty}\bar{\boldsymbol{\mathsf{1}}}_{\:J}\Big(\mathbf{X},\mathbf{a}_{\:\text{-}t}\big(\hat{m}_{\:t,\:J}(\mathbf{X},\mathbf{a},\mathbf{b})\big),\mathbf{b}\Big) \\
&\hspace{4cm}\overset{a.s.}{=}\lim_{J\to\infty}\lim_{|\mathbf{X}|\to\infty}\bar{\boldsymbol{\mathsf{1}}}\Big(\mathbf{X},\mathbf{a}_{\:\text{-}t}\big(\hat{m}_{\:t,\:J}(\mathbf{X},\mathbf{a},\mathbf{b})\big),\mathbf{b}\Big)\:[P]\nonumber.
\end{align}

Now, from the definition of $\hat{m}_{\:t,\:J}(\mathbf{X},\mathbf{a},\mathbf{b})$ we have:
\begin{align}\label{main:proof:eq:c}
\bar{\boldsymbol{\mathsf{1}}}_{\:J}\Big(\mathbf{X}
,\mathbf{a}_{\:\text{-}t}\big(\hat{m}_{\:t,\:J}(\mathbf{X},\mathbf{a},\mathbf{b})\big),\mathbf{b}\Big)=\bar{\boldsymbol{\mathsf{1}}}_{\:J}\Big(\mathbf{X},\mathbf{a},\mathbf{b}_{\:\text{-}t}\big(\hat{m}_{\:t,\:J}(\mathbf{X},\mathbf{a},\mathbf{b})\big)\Big).
\end{align}

Equation \ref{main:proof:eq:c} ensures equality of the left-hand sides of equations \ref{main:proof:eq:a} and \ref{main:proof:eq:b}. So, we can equate the right hand sides of equations \ref{main:proof:eq:a} and \ref{main:proof:eq:b} to conclude:
\begin{align*}
&\lim_{J\to\infty}\lim_{|\mathbf{X}|\to\infty}\bar{\boldsymbol{\mathsf{1}}}\Big(\mathbf{X},\mathbf{a},\mathbf{b}_{\:\text{-}t}\big(\hat{m}_{\:t,\:J}(\mathbf{X},\mathbf{a},\mathbf{b})\big)\Big) \\
&\hspace{4cm}\overset{a.s.}{=}\lim_{J\to\infty}\lim_{|\mathbf{X}|\to\infty}\bar{\boldsymbol{\mathsf{1}}}\Big(\mathbf{X},\mathbf{a}_{\:\text{-}t}\big(\hat{m}_{\:t,\:J}(\mathbf{X},\mathbf{a},\mathbf{b})\big),\mathbf{b})\Big).
\end{align*}
\end{proof} 

\subsection{Proof of Theorem \ref{thm:n_l=n/2}}\label{pf:n_l=n/2}

\begin{proof}
Let $\epsilon>0$, and note that
\begin{align*}
&(\mathbf{a},\mathbf{b})=\Big(\mathbf{a},\mathbf{b}_{\:\text{-}t}\big(\hat{m}_{\:t,\:J}(\mathbf{X},\mathbf{a},\mathbf{b})\big)\Big)\cup\Big(\mathbf{a}_{\:\text{-}t}\big(\hat{m}_{\:t,\:J}(\mathbf{X},\mathbf{a},\mathbf{b})\big),\mathbf{b}\Big) \\
&\hspace{6cm}\cup\mathcal{A}_{\:t}\big(\hat{m}_{\:t,\:J}(\mathbf{X},\mathbf{a},\mathbf{b}),\mathbf{a},\mathbf{b}\big).
\end{align*}

In other words:
\begin{multline*}
\mathsf{1}\big(\mathbf{x}_{\:i}\in(\mathbf{a},\mathbf{b})\big)=\mathsf{1}\bigg(\mathbf{x}_{\:i}\in\Big(\mathbf{a},\mathbf{b}_{\:\text{-}t}\big(\hat{m}_{\:t,\:J}(\mathbf{X},\mathbf{a},\mathbf{b})\big)\Big)\bigg) \\
+\mathsf{1}\bigg(\mathbf{x}_{\:i}\in\Big(\mathbf{a}_{\:\text{-}t}\big(\hat{m}_{\:t,\:J}(\mathbf{X},\mathbf{a},\mathbf{b})\big),\mathbf{b}\Big)\bigg)+\mathsf{1}\Big(\mathbf{x}_{\:i}\in\mathcal{A}_{\:t}\big(\hat{m}_{\:t,\:J}(\mathbf{X},\mathbf{a},\mathbf{b}),\mathbf{a},\mathbf{b}\big)\Big).
\end{multline*}

Summing over $i\in I$, and dividing both sides by $|\mathbf{X}|$, we obtain:
\begin{multline}\label{main:proof:eq:d}
\bar{\boldsymbol{\mathsf{1}}}(\mathbf{X},\mathbf{a},\mathbf{b})=\bar{\boldsymbol{\mathsf{1}}}\Big(\mathbf{X},\mathbf{a},\mathbf{b}_{\:\text{-}t}\big(\hat{m}_{\:t,\:J}(\mathbf{X},\mathbf{a},\mathbf{b})\big)\Big) +\bar{\boldsymbol{\mathsf{1}}}\Big(\mathbf{X},\mathbf{a}_{\:\text{-}t}\big(\hat{m}_{\:t,\:J}(\mathbf{X},\mathbf{a},\mathbf{b})\big),\mathbf{b}\Big) \\
+\frac{1}{|\mathbf{X}|}\sum_{i\in I}\mathbf{\mathsf{1}}\Big(\mathbf{x}_{\:i}\in\mathcal{A}_{\:t}\big(\hat{m}_{\:t,\:J}(\mathbf{X},\mathbf{a},\mathbf{b}),\mathbf{a},\mathbf{b}\big)\Big).
\end{multline}

Now, $\mathcal{A}_{\:t}\big(\hat{m}_{\:t,\:J}(\mathbf{X},\mathbf{a},\mathbf{b}),\mathbf{a},\mathbf{b}\big)\in\mathcal{N}_{\:p}$ for any $(\mathbf{a},\mathbf{b})\subseteq(\mathbf{0},\mathbf{1})$. So, by Assumption \ref{assumption:sup_null_set} there exists a null set $\mathbb{M}_{\:1}\subset\Omega$, such that, for each $w\in\Omega\setminus\mathbb{M}_{\:1}$, there exists an $n_{\:w,\:2}\in\mathbb{N}$, such that if $n>n_{\:w,\:2}$, then for any $(\mathbf{a},\mathbf{b})\subseteq(\mathbf{0},\mathbf{1})$, we have:
\begin{align*}
\bigg|\:\frac{1}{|\mathbf{X}(w,n)|}\sum_{i\in I}\mathbf{\mathsf{1}}\Big(\mathbf{x}_{\:i}(w)\in\mathcal{A}_{\:t}\big(\hat{m}_{\:t,\:J}(\mathbf{X}(w,n),\mathbf{a},\mathbf{b}),\mathbf{a},\mathbf{b}\big)\Big)\:\bigg|<\frac{2\epsilon}{3}.
\end{align*}

Then, from equation \ref{main:proof:eq:d}, we have for $w\in\Omega\setminus\mathbb{M}_{\:1}$ and $n>n_{\:w,\:2}$:
\begin{multline}\label{main:proof:eq:e}
\Big|\:\bar{\boldsymbol{\mathsf{1}}}\big(\mathbf{X}(w,n),\mathbf{a},\mathbf{b}\big)-\bar{\boldsymbol{\mathsf{1}}}\Big(\mathbf{X}(w,n),\mathbf{a},\mathbf{b}_{\:\text{-}t}\big(\hat{m}_{\:t,\:J}(\mathbf{X}(w,n),\mathbf{a},\mathbf{b})\big)\Big) \\
-\bar{\boldsymbol{\mathsf{1}}}\Big(\mathbf{X}(w,n),\mathbf{a}_{\:\text{-}t}\big(\hat{m}_{\:t,\:J}(\mathbf{X}(w,n),\mathbf{a},\mathbf{b})\big),\mathbf{b}\Big)\:\Big|<\frac{2\epsilon}{3}.
\end{multline}

From the definition of $\hat{m}_{\:t,\:J}\big(\mathbf{X}(w,n),\mathbf{a},\mathbf{b}\big)$, we have:
\begin{multline}\label{main:proof:eq:f}
\bar{\boldsymbol{\mathsf{1}}}\Big(\mathbf{X}(w,n),\mathbf{a},\mathbf{b}_{\:\text{-}t}\big(\hat{m}_{\:t,\:J}(\mathbf{X}(w,n),\mathbf{a},\mathbf{b})\big)\Big) \\
=\bar{\boldsymbol{\mathsf{1}}}\Big(\mathbf{X}(w,n),\mathbf{a}_{\:\text{-}t}\big(\hat{m}_{\:t,\:J}(\mathbf{X}(w,n),\mathbf{a},\mathbf{b})\big),\mathbf{b}\Big).
\end{multline}

From Theorem \ref{thm:n_l=n_r} part (B), we know that there exists an integer $J_{\:\epsilon,\:P}$ such that for $J>J_{\:\epsilon,\:P}$, there exists integer $n_{\:w,\:J}$, such that if $n>n_{\:w,\:J}$, then $\bar{\mathcal{E}}_{\:J}\big(\mathbf{X}(w,n),\mathbf{a},\mathbf{b}\big)<\epsilon$, so that $|\bar{\boldsymbol{\mathsf{1}}}_{\:J}\big(\mathbf{X}(w,n),\mathbf{a},\mathbf{b}\big)-\bar{\boldsymbol{\mathsf{1}}}\big(\mathbf{X}(w,n),\mathbf{a},\mathbf{b}\big)|<\epsilon$ for $(\mathbf{a},\mathbf{b})\in(\mathbf{0},\mathbf{1})$. In particular, for $J>J_{\:w,\:J}$ and $n>n_{\:w,\:J}$ we have:
\begin{align}\label{main:proof:eq:g}
&\Big|\:\bar{\boldsymbol{\mathsf{1}}}\Big(\mathbf{X}(w,n),\mathbf{a},\mathbf{b}_{\:\text{-}t}\big(\hat{m}_{\:t,\:J}(\mathbf{X}(w,n),\mathbf{a},\mathbf{b})\big)\Big) \nonumber \\ &\hspace{3cm}-\bar{\boldsymbol{\mathsf{1}}}_{\:J}\Big(\mathbf{X}(w,n),\mathbf{a},\mathbf{b}_{\:\text{-}t}\big(\hat{m}_{\:t,\:J}(\mathbf{X}(w,n),\mathbf{a},\mathbf{b})\big)\Big)\:\Big|<\frac{2\epsilon}{3}.
\end{align}
We likewise have:
\begin{align}\label{main:proof:eq:h}
&\Big|\:\bar{\boldsymbol{\mathsf{1}}}\Big(\mathbf{X}(w,n),\mathbf{a}_{\:\text{-}t}\big(\hat{m}_{\:t,\:J}(\mathbf{X}(w,n),\mathbf{a},\mathbf{b})\big),\mathbf{b}\Big) \nonumber \\ &\hspace{3cm}-\bar{\boldsymbol{\mathsf{1}}}_{\:J}\Big(\mathbf{X}(w,n),\mathbf{a}_{\:\text{-}t}\big(\hat{m}_{\:t,\:J}(\mathbf{X}(w,n),\mathbf{a},\mathbf{b})\big),\mathbf{b}\Big)\:\Big|<\frac{2\epsilon}{3}.
\end{align}

Then, from equations \ref{main:proof:eq:e}--\ref{main:proof:eq:h}, we have:
\begin{align*}
&\bigg|\:\bar{\boldsymbol{\mathsf{1}}}\big(\mathbf{X}(w,n),\mathbf{a},\mathbf{b}\big)-2\bar{\boldsymbol{\mathsf{1}}}\bigg(\mathbf{X}(w,n),\mathbf{a},\mathbf{b}_{\:\text{-}t}\Big(\hat{m}_{\:t,\:J}\big(\mathbf{X}(w,n),\mathbf{a},\mathbf{b}\big)\Big)\bigg)\:\bigg| \\
&=\Bigg|\:\Bigg(\bar{\boldsymbol{\mathsf{1}}}\big(\mathbf{X}(w,n),\mathbf{a},\mathbf{b}\big)-\bar{\boldsymbol{\mathsf{1}}}\bigg(\mathbf{X}(w,n),\mathbf{a},\mathbf{b}_{\:\text{-}t}\Big(\hat{m}_{\:t,\:J}\big(\mathbf{X}(w,n),\mathbf{a},\mathbf{b}\big)\Big)\bigg) \\
&\hspace{4cm}-\bar{\boldsymbol{\mathsf{1}}}\bigg(\mathbf{X}(w,n),\mathbf{a}_{\:\text{-}t}\Big(\hat{m}_{\:t,\:J}\big(\mathbf{X}(w,n),\mathbf{a},\mathbf{b}\big)\Big),\mathbf{b}\bigg)\Bigg) \\
&-\Bigg(\bar{\boldsymbol{\mathsf{1}}}\bigg(\mathbf{X}(w,n),\mathbf{a},\mathbf{b}_{\:\text{-}t}\Big(\hat{m}_{\:t,\:J}\big(\mathbf{X}(w,n),\mathbf{a},\mathbf{b}\big)\Big)\bigg) \\
&\hspace{4cm}-\bar{\boldsymbol{\mathsf{1}}}_{\:J}\bigg(\mathbf{X}(w,n),\mathbf{a},\mathbf{b}_{\:\text{-}t}\Big(\hat{m}_{\:t,\:J}\big(\mathbf{X}(w,n),\mathbf{a},\mathbf{b}\big)\Big)\bigg)\Bigg) \\
&+\Bigg(\bar{\boldsymbol{\mathsf{1}}}\bigg(\mathbf{X}(w,n),\mathbf{a}_{\:\text{-}t}\Big(\hat{m}_{\:t,\:J}\big(\mathbf{X}(w,n),\mathbf{a},\mathbf{b}\big)\Big),\mathbf{b}\bigg) \\ &\hspace{4cm}-\bar{\boldsymbol{\mathsf{1}}}_{\:J}\bigg(\mathbf{X}(w,n),\mathbf{a}_{\:\text{-}t}\Big(\hat{m}_{\:t,\:J}\big(\mathbf{X}(w,n),\mathbf{a},\mathbf{b}\big)\Big),\mathbf{b}\bigg)\Bigg)
\end{align*}
\begin{align*}
&+\Bigg(\bar{\boldsymbol{\mathsf{1}}}_{\:J}\bigg(\mathbf{X}(w,n),\mathbf{a},\mathbf{b}_{\:\text{-}t}\Big(\hat{m}_{\:t,\:J}\big(\mathbf{X}(w,n),\mathbf{a},\mathbf{b}\big)\Big)\bigg) \\
&\hspace{4cm}-\bar{\boldsymbol{\mathsf{1}}}_{\:J}\bigg(\mathbf{X}(w,n),\mathbf{a}_{\:\text{-}t}\Big(\hat{m}_{\:t,\:J}\big(\mathbf{X}(w,n),\mathbf{a},\mathbf{b}\big)\Big),\mathbf{b}\bigg)\Bigg)\:\Bigg| 
\end{align*}
\begin{align*}
&\leq\bigg|\:\bar{\boldsymbol{\mathsf{1}}}\big(\mathbf{X}(w,n),\mathbf{a},\mathbf{b}\big)-\bar{\boldsymbol{\mathsf{1}}}\bigg(\mathbf{X}(w,n),\mathbf{a},\mathbf{b}_{\:\text{-}t}\Big(\hat{m}_{\:t,\:J}\big(\mathbf{X}(w,n),\mathbf{a},\mathbf{b}\big)\Big)\bigg) \\
&\hspace{4cm}-\bar{\boldsymbol{\mathsf{1}}}\bigg(\mathbf{X}(w,n),\mathbf{a}_{\:\text{-}t}\Big(\hat{m}_{\:t,\:J}\big(\mathbf{X}(w,n),\mathbf{a},\mathbf{b}\big)\Big),\mathbf{b}\bigg)\:\bigg| \\
&+\bigg|\:\bar{\boldsymbol{\mathsf{1}}}\bigg(\mathbf{X}(w,n),\mathbf{a},\mathbf{b}_{\:\text{-}t}\Big(\hat{m}_{\:t,\:J}\big(\mathbf{X}(w,n),\mathbf{a},\mathbf{b}\big)\Big)\bigg) \\
&\hspace{4cm}-\bar{\boldsymbol{\mathsf{1}}}_{\:J}\bigg(\mathbf{X}(w,n),\mathbf{a},\mathbf{b}_{\:\text{-}t}\Big(\hat{m}_{\:t,\:J}\big(\mathbf{X}(w,n),\mathbf{a},\mathbf{b}\big)\Big)\bigg)\:\bigg| \\
&+\bigg|\:\bar{\boldsymbol{\mathsf{1}}}\bigg(\mathbf{X}(w,n),\mathbf{a}_{\:\text{-}t}\Big(\hat{m}_{\:t,\:J}\big(\mathbf{X}(w,n),\mathbf{a},\mathbf{b}\big)\Big),\mathbf{b}\bigg) \\
&\hspace{4cm}-\bar{\boldsymbol{\mathsf{1}}}_{\:J}\bigg(\mathbf{X}(w,n),\mathbf{a}_{\:\text{-}t}\Big(\hat{m}_{\:t,\:J}\big(\mathbf{X}(w,n),\mathbf{a},\mathbf{b}\big)\Big),\mathbf{b}\bigg)\:\bigg| \\
&<\frac{2\epsilon}{3}+\frac{2\epsilon}{3}+\frac{2\epsilon}{3}=2\epsilon \\
&\Leftrightarrow\bigg|\:\bar{\boldsymbol{\mathsf{1}}}\bigg(\mathbf{X}(w,n),\mathbf{a},\mathbf{b}_{\:\text{-}t}\Big(\hat{m}_{\:t,\:J}\big(\mathbf{X}(w,n),\mathbf{a},\mathbf{b}\big)\Big)\bigg)-\frac{1}{2}\bar{\boldsymbol{\mathsf{1}}}\big(\mathbf{X}(w,n),\mathbf{a},\mathbf{b}\big)\:\bigg|<\epsilon.
\end{align*}

Since $\epsilon$ is arbitrary, we must have:
\begin{align*}
\lim_{J\to\infty}\lim_{|\mathbf{X}|\to\infty}\bar{\boldsymbol{\mathsf{1}}}\Big(\mathbf{X},\mathbf{a},\mathbf{b}_{\:\text{-}t}\big(\hat{m}_{\:t,\:J}(\mathbf{X},\mathbf{a},\mathbf{b})\big)\Big)\overset{a.s.}{=}\frac{1}{2}\lim_{J\to\infty}\lim_{|\mathbf{X}|\to\infty}\bar{\boldsymbol{\mathsf{1}}}\big(\mathbf{X},\mathbf{a},\mathbf{b})\big)\:[P].
\end{align*}

Similarly, we can also prove:
\begin{align*}
\lim_{J\to\infty}\lim_{|\mathbf{X}|\to\infty}\bar{\boldsymbol{\mathsf{1}}}\Big(\mathbf{X},\mathbf{a}_{\:\text{-}t}\big(\hat{m}_{\:t,\:J}(\mathbf{X},\mathbf{a},\mathbf{b}),\mathbf{b}\big)\Big)\overset{a.s.}{=}\frac{1}{2}\lim_{J\to\infty}\lim_{|\mathbf{X}|\to\infty}\bar{\boldsymbol{\mathsf{1}}}\big(\mathbf{X},\mathbf{a},\mathbf{b})\big)\:[P].
\end{align*}
\end{proof}

\section{Implementing the proposed method of construction of a balanced $k$-d tree in practice}\label{ap:construct}

\subsection{Algorithmic approach}\label{sec:eff}

To construct a balanced $k$-d tree for distributed data as described in section \ref{sec:algorithm}, the set of statistics $\{\{\bar{\boldsymbol{\mathsf{C}}}_{\:\boldsymbol{j}}(\mathbf{X})\}\}_{\boldsymbol{j}\in\mathbb{N}_{\:J\:,p}}$ can be computed in a single step, using a programming model such as map-reduce (fully detailed below). Then, equations \ref{kd:eq:a} and \ref{kd:eq:b} can be used recursively to compute $\hat{m}_{d,k}^J(\mathbf{X})$ for $1\leq d\leq D,1\leq k\leq 2^{d-1}$, resulting in an approximately balanced $k$-d tree. Because $\boldsymbol{\mathsf{C}}_{\:\boldsymbol{j}}(\mathbf{X})=\sum_{i\in I}\boldsymbol{\mathsf{c}}_{\:\boldsymbol{j}}(\mathbf{x}_{\:i})$, where $\boldsymbol{\mathsf{c}}_{\:\boldsymbol{j}}(\mathbf{x}_{\:i})=\prod_{l=1}^p\mathsf{c}^{\:j_{\:l}}(x_{\:i_{\:l}})$, this implies the computation of $2Jp$ trigonometric terms for each $\mathbf{x}_{\:i}$. 

Instead, let us consider $K$ parameters $J_{\:1},J_{\:2},\dots,J_{\:K}$, such that $J=\prod_{k=1}^KJ_{\:k}$. Denote $\underline{\boldsymbol{J}}=\big(J_{\:1},J_{\:2},\dots,J_{\:K}\big)$, and define index sets
\begin{align*}
&\mathbb{N}_{\:\underline{\boldsymbol{J}}}=\{0\}\cup\big(\{1,\dots,2J_{\:1}\}\times\{1,\dots,J_{\:2}\}\times\dots\times\{1,\dots,J_{\:K}\}\big), \\
&\mathbb{N}_{\:\underline{\boldsymbol{J}},\:p}=\mathbb{N}_{\:\underline{\boldsymbol{J}}}\times\dots\times\mathbb{N}_{\:\underline{\boldsymbol{J}}} (p\text{ times }).
\end{align*}
Note that a general element $\underline{\underline{\boldsymbol{j}}}\in\mathbb{N}_{\:\underline{\boldsymbol{J}},\:p}$ is a $p$-tuple $\big(\underline{\boldsymbol{j}}_{\:1},\dots,\underline{\boldsymbol{j}}_{\:p}\big)$, where for $1\leq l\leq p$, the entry $\underline{\boldsymbol{j}}_{\:l}$ is either $0$, or a $K$-length vector $\big(j_{\:l,\:1},\dots,j_{\:l,\:K}\big)$, where $1\leq j_{\:l,\:1}\leq 2J_{\:1}$ and $1\leq j_{\:l,\:k}\leq J_{\:k}$ for $2\leq k\leq K$. With these notations in place, we define the following functions:
\begin{align*}
&\mathsf{c}^{\:2j-1}(z)=\cos^{2j-1}(z),\quad \mathsf{c}^{\:2j}(z)=\sin(z)\cdot\cos^{2j-2}(z) \\
&\hspace{5cm}\text{ for }j\in\mathbb{N}^+\text{ and }z\in(0,1); \\
&\dot{\mathsf{c}}^{\:j}(z,J')=\cos^{j-1}(2J'z) \\
&\hspace{5cm}\text{ for }j\in\mathbb{N}^+,z\in(0,1)\text{ and }J'\in\mathbb{N}^+; \\
&\tilde{\boldsymbol{\mathsf{c}}}^{\:\underline{\underline{\boldsymbol{j}}}}\big(\mathbf{x},\underline{\boldsymbol{J}}\big)=\prod_{l=1}^p\Big(\mathsf{c}^{\:j_{\:l,\:1}}(x_{\:l})\cdot \prod_{k=2}^K\dot{\mathsf{c}}^{\:j_{\:l,\:k}}(x_{\:l},L_{\:k-1})\Big)^{\mathsf{1}\big(\underline{\boldsymbol{j}}_{\:l}\neq 0\big)} \\
&\hspace{5cm}\text{ for }\mathbf{x}\in(\mathbf{0},\mathbf{1}),\underline{\underline{\boldsymbol{j}}}\in\mathbb{N}_{\:\underline{\boldsymbol{J}},\:p}\text{ and }\underline{\boldsymbol{J}}\in{\mathbb{N}^+}^K \\
&\text{ where }L_{\:k}=\prod_{k'=1}^kJ_{\:k}\text{ for }1\leq k<K,\text{ here }\underline{\boldsymbol{J}}=\big(J_{\:1},J_{\:2},\dots,J_{\:K}\big).
\end{align*}
Given $\underline{\boldsymbol{J}}\in{\mathbb{N}^+}^K$, we define the statistic: $\tilde{\boldsymbol{\mathsf{C}}}^{\:\underline{\underline{\boldsymbol{j}}}}\big(\mathbf{X},\underline{\boldsymbol{J}}\big)=\sum_{i\in I}\tilde{\boldsymbol{\mathsf{c}}}^{\:\underline{\underline{\boldsymbol{j}}}}\big(\mathbf{x}_{\:i},\underline{\boldsymbol{J}}\big)$ and its standardized version $\bar{\tilde{\boldsymbol{\mathsf{C}}}}^{\:\underline{\underline{\boldsymbol{j}}}}\big(\mathbf{X},\underline{\boldsymbol{J}}\big)=\big(\nicefrac{1}{|\mathbf{X}|}\big)\cdot\tilde{\boldsymbol{\mathsf{C}}}^{\:\underline{\underline{\boldsymbol{j}}}}\big(\mathbf{X},\underline{\boldsymbol{J}}\big)$ for $\underline{\underline{\boldsymbol{j}}}\in\mathbb{N}_{\:\underline{\boldsymbol{J}},\:p}$. Finally, define the sets of statistics
\begin{align*}
&\bar{\boldsymbol{\mathsf{C}}}_{{\:(0:2J)}^{\:p}}(\mathbf{X})=\{\{\bar{\boldsymbol{\mathsf{C}}}_{\:\boldsymbol{j}}(\mathbf{X})\}\}_{\:\boldsymbol{j}\in\mathbb{N}_{\:J;\:p}}, \\
&\bar{\tilde{\boldsymbol{\mathsf{C}}}}^{{\:0;\:(1:2J_{\:1}),\:(1:J_{\:2}),\:\dots,\:(1:J_{\:K})}^{\:p}}(\mathbf{X},\underline{\boldsymbol{J}})=\{\{\bar{\tilde{\boldsymbol{\mathsf{C}}}}^{\:\underline{\underline{\boldsymbol{j}}}}(\mathbf{X},\boldsymbol{J})\}\}_{\:\underline{\underline{\boldsymbol{j}}}\in\mathbb{N}_{\:\underline{\boldsymbol{J}};\:p}}.
\end{align*}

We then have the following theorem, which asserts the existence of a transformation recovering our original basic construction.

\begin{theorem}\label{chev:thm:main}
For $K\in\mathbb{N}^+$ and $\underline{\boldsymbol{J}}\in{\mathbb{N}^+}^K$, there exists a linear transformation $\mathcal{T}^{(K,p)}_{\:\underline{\boldsymbol{J}}}$ such that
\begin{align*}
\bar{\boldsymbol{\mathsf{C}}}_{{\:(0:2J)}^{\:p}}(X)=\mathcal{T}^{(K,p)}_{\:\underline{\boldsymbol{J}}}\Big(\bar{\tilde{\boldsymbol{\mathsf{C}}}}^{{\:0;\:(1:2J_{\:1}),\:(1:J_{\:2}),\:\dots,\:(1:J_{\:K})}^{\:p}}(X,\underline{\boldsymbol{J}})\Big).
\end{align*}
\end{theorem}

\begin{proof}
See Appendix \ref{pf:chev:thm:main}.
\end{proof}

For any arbitrary partition $\big\{\mathbf{X}_{\:1},\dots,\mathbf{X}_{\:R}\big\}$ of the data $\mathbf{X}$ into $R$ subsets (with $R\in\mathbb{N}^+$ and $1\leq R\leq n$), once again we have the relation  $\tilde{\boldsymbol{\mathsf{C}}}^{\:\underline{\underline{\boldsymbol{j}}}}\big(\mathbf{X},\underline{\boldsymbol{J}}\big)=\sum_{r=1}^R\tilde{\boldsymbol{\mathsf{C}}}^{\:\underline{\underline{\boldsymbol{j}}}}\big(\mathbf{X}_{\:r},\underline{\boldsymbol{J}}\big)$ for $\:\underline{\underline{\boldsymbol{j}}}\in\mathbb{N}_{\:\underline{\boldsymbol{J}};\:p}$. Hence for any distributed data set $\mathbf{X}$, the collection of statistics $\tilde{\boldsymbol{\mathsf{C}}}^{\:\underline{\underline{\boldsymbol{j}}}}\big(\mathbf{X},\underline{\boldsymbol{J}}\big)$ for $\underline{\underline{\boldsymbol{j}}}\in\mathbb{N}_{\:\underline{\boldsymbol{J}};\:p}$, can exactly be computed in parallel. Observe that $\boldsymbol{0}\in\mathbb{N}_{\:\underline{\boldsymbol{J}};\:p}$ and  $\tilde{\boldsymbol{\mathsf{C}}}^{\:\boldsymbol{0}}(\mathbf{X})=\sum_{i\in I}1=|\mathbf{X}|$. Then, for $\underline{\underline{\boldsymbol{j}}}\in\mathbb{N}_{\:\underline{\boldsymbol{J}};\:p}$, we have that $\bar{\tilde{\boldsymbol{\mathsf{C}}}}^{\:\underline{\underline{\boldsymbol{j}}}}\big(\mathbf{X},\underline{\boldsymbol{J}}\big)=\big(\nicefrac{1}{\mathbf{X}}\big)\cdot\tilde{\boldsymbol{\mathsf{C}}}^{\:\underline{\underline{\boldsymbol{j}}}}\big(\mathbf{X},\underline{\boldsymbol{J}}\big)=\tilde{\boldsymbol{\mathsf{C}}}^{\:\underline{\underline{\boldsymbol{j}}}}\big(\mathbf{X},\underline{\boldsymbol{J}}\big)/\tilde{\boldsymbol{\mathsf{C}}}^{\:\boldsymbol{0}}\big(\mathbf{X},\underline{\boldsymbol{J}}\big)$. Thus, the collection of standardized statistics $\bar{\tilde{\boldsymbol{\mathsf{C}}}}^{\:\underline{\underline{\boldsymbol{j}}}}\big(\mathbf{X},\underline{\boldsymbol{J}}\big)$ for $\underline{\underline{\boldsymbol{j}}}\in\mathbb{N}_{\:\underline{\boldsymbol{J}};\:p}$ also can be exactly computed in parallel.

As a result, it is possible to reduce substantially the number of trigonometric terms needed per data point: First, the collection of statistics $\bar{\tilde{\boldsymbol{\mathsf{C}}}}^{\:\underline{\underline{\boldsymbol{j}}}}\big(\mathbf{X},\underline{\boldsymbol{J}}\big)$ for $\underline{\underline{\boldsymbol{j}}}\in\mathbb{N}_{\:\underline{\boldsymbol{J}};\:p}$ is computed in a single step. Then, the transformation $\mathcal{T}^{(K,p)}_{\:\underline{\boldsymbol{J}}}$ from Theorem \ref{chev:thm:main} is used to compute $\bar{\boldsymbol{\mathsf{C}}}_{\:\boldsymbol{j}}(\mathbf{X})$ for $\boldsymbol{j}\in\mathbb{N}_{\:J\:,p}$. Then, as before, equations \ref{kd:eq:a} and \ref{kd:eq:b} are used recursively to compute $\hat{m}_{\:d,\:k,\:J}(\mathbf{X})$ for $1\leq d\leq D,1\leq k\leq 2^{d-1}$, to construct the entire approximate $k$-d tree. 

While $(2J+1)^p$ statistics must still be computed, only $p\cdot K$ trigonometric terms are needed per data point (using the relation $\sin(x)=\sqrt{1-\cos^2(x)}$). Hence when $K<<J$, as supported by the simulation results of section  \ref{sec:performance}, this yields a substantial reduction in computational overhead, replacing serial for-loop trigonometric computations with serial for-loop multiplication by constant terms.

\subsection{Map-reduce implementation}\label{sec:MapReduce}

We now describe how to compute the set of statistics necessary for $k$-d tree construction; i.e.,  $\tilde{\boldsymbol{\mathsf{C}}}^{\:\underline{\underline{\boldsymbol{j}}}}(\mathbf{X},\underline{\boldsymbol{J}})$ for $\underline{\underline{\boldsymbol{j}}}\in\mathbb{N}_{\:\underline{\boldsymbol{J}};\:p}$. Suppose the user specifies a partition $\{I_{\:1},\ldots,I_{\:R}\}$ of the index set $I$, so that the data set of interest $\mathbf{X}$ is distributed according to corresponding subsets $\mathbf{X}_{\:1},\ldots,\mathbf{X}_{\:R}$. The data are now referenced through the set of key-value pairs $\{1,\mathbf{X}_{\:1}\},\ldots,\{R,\mathbf{X}_{\:R}\}$, which in turn serve as an input to the map step. Observe that $\tilde{\boldsymbol{\mathsf{C}}}^{\:\underline{\underline{\boldsymbol{j}}}}(\mathbf{X}_{\:r},\underline{\boldsymbol{J}})=\sum_{i\in I_{\:r}}\tilde{\boldsymbol{\mathsf{c}}}^{\:\underline{\underline{\boldsymbol{j}}}}(\mathbf{x}_{\:i},\underline{\boldsymbol{J}})$ for $\underline{\underline{\boldsymbol{j}}}\in\mathbb{N}_{\:\underline{\boldsymbol{J}};\:p}$. Hence, given an arbitrary key-value pair $\{r,\mathbf{X}_{\:r}\}$, the terms $\tilde{\boldsymbol{\mathsf{c}}}^{\:\underline{\underline{\boldsymbol{j}}}}(\mathbf{x}_{\:i},\underline{\boldsymbol{J}})$ for $\underline{\underline{\boldsymbol{j}}}\in\mathbb{N}_{\:\underline{\boldsymbol{J}};\:p}$ will be computed, and then summed over $i\in I_{\:r}$, to obtain $\tilde{\boldsymbol{\mathsf{C}}}^{\:\underline{\underline{\boldsymbol{j}}}}(\mathbf{X}_{\:r},\underline{\boldsymbol{J}})$. Thus, for the input key-value pair $\{r,\mathbf{X}_{\:r}\}$, the corresponding map step output is the set of intermediate key-value pairs  $\{\underline{\underline{\boldsymbol{j}}},\tilde{\boldsymbol{\mathsf{C}}}^{\:\underline{\underline{\boldsymbol{j}}}}(\mathbf{X}_{\:r},\underline{\boldsymbol{J}})\}$ for $\underline{\underline{\boldsymbol{j}}}\in\mathbb{N}_{\:\underline{\boldsymbol{J}};\:p}$. 

In the reduce step, we take the reduce function to be the addition operator. So, for each key, we sum over the corresponding values from the set of all intermediate key-value pairs. Since $I$ is the disjoint union of $I_{\:1},\ldots,I_{\:R}$, the final map-reduce output is the set of key-value pairs $\{\underline{\underline{\boldsymbol{j}}},\tilde{\boldsymbol{\mathsf{C}}}^{\:\underline{\underline{\boldsymbol{j}}}}(\mathbf{X},\underline{\boldsymbol{J}})\}$ for $\underline{\underline{\boldsymbol{j}}}\in\mathbb{N}_{\:\underline{\boldsymbol{J}};\:p}$. As described by \citet[section~3.2]{Chakravorty2021}, the partition $\{I_{\:1},\ldots,I_{\:R}\}$ may also be determined implicitly through an architecture such as Spark \citep{Zaharia2016}. In this case, $X$ is treated as a resilient distributed data set (RDD), via a function $\lambda(\mathbf{x})=\{\{\tilde{\boldsymbol{\mathsf{c}}}^{\:\underline{\underline{\boldsymbol{j}}}}(\mathbf{x},\underline{\boldsymbol{J}})\}\}_{\:\underline{\underline{\boldsymbol{j}}}\in\mathbb{N}_{\:\underline{\boldsymbol{J}};\:p}}$ used to transform it to an intermediate RDD by way of a flat-map transformation. This intermediate RDD is then further transformed, once again taking the reduce function to be the addition operator.

After implementing the overall map-reduce step described above, the resulting statistics $\tilde{\boldsymbol{\mathsf{C}}}^{\:\underline{\underline{\boldsymbol{j}}}}(\mathbf{X},\underline{\boldsymbol{J}})$ for $\underline{\underline{\boldsymbol{j}}}\in\mathbb{N}_{\:\underline{\boldsymbol{J}};\:p}$ may now be used to compute $\hat{m}_{\:d,\:k}^J(\mathbf{X})$ for $1\leq d\leq D,1\leq k\leq 2^{d-1}$, enabling construction of the $k$-d tree.

\subsection{Proof of Theorem \ref{chev:thm:main}}\label{pf:chev:thm:main}

To establish the existence of the linear transformation asserted by Theorem \ref{chev:thm:main}, we require some preliminary results. To this end, let us first define the quantities
\begin{align*}
&\boldsymbol{\mathsf{c}}_{\:(1:2J)}(z)=\big(c_{\:1}(z),\dots,c_{\:2J}(z)\big); \\
&\boldsymbol{\mathsf{c}}_{\:(0:2J)}(z)=\big(c_{\:0}(z),c_{\:1}(z),\dots,c_{\:2J}(z)\big)=\big(1,\boldsymbol{\mathsf{c}}_{\:(1:2J)}(z)\big).
\end{align*}
Here $\boldsymbol{\mathsf{c}}_{\:(0:2J)}(z)$ is indexed by the set $\mathbb{N}_{\:J}=\{0,1,\dots,2J\}$. Now, for $\underline{\boldsymbol{J}}=(J_{\:1},\dots,J_{\:K})$, define the index set
\begin{align*}
\mathbb{N}_{\:\{\underline{\boldsymbol{J}}\}}=\{1,\dots,2J_{\:1}\}\times\{1,\dots,J_{\:2}\}\times\dots\times\{1,\dots,J_{\:K}\}.
\end{align*}
We then add the element $0$ to obtain the index set
\begin{align*}
\mathbb{N}_{\:\underline{\boldsymbol{J}}}=\{0\}\cup\big(\{1,\dots,2J_{\:1}\}\times\{1,\dots,J_{\:2}\}\times\dots\times\{1,\dots,J_{\:K}\}\big).
\end{align*}
Also, for the new index element $0$, define: $\tilde{c}^{\:0}(z,\underline{\boldsymbol{J}})=1$. Now we define:
\begin{align*}
\tilde{\boldsymbol{\mathsf{c}}}^{\:0;\:(1:2J_{\:1}),\:(1:J_{\:2}),\:\dots\:,\:(1:J_{\:K})}(z,\underline{\boldsymbol{J}})=\{\{\tilde{\mathsf{c}}^{\:\underline{\boldsymbol{j}}}(z,\underline{\boldsymbol{J}})\}\}_{\:\big\{\underline{\boldsymbol{j}}\in\mathbb{N}_{\:\underline{\boldsymbol{J}}}\big\}}.
\end{align*}

Then, we have the following lemma.

\begin{lemma}\label{chev:lem:gen:1d}
For $K\in\mathbb{N}^+$ and $\underline{\boldsymbol{J}}\in{\mathbb{N}^+}^K$, there exists a linear transformation $\mathcal{T}^{(K)}_{\:\underline{\boldsymbol{J}}}$ such that:
\begin{align*}
\boldsymbol{\mathsf{c}}_{\:(0:2J)}(z)=\mathcal{T}^{(K)}_{\:\underline{\boldsymbol{J}}}\big(\tilde{\boldsymbol{\mathsf{c}}}^{\:0;\:(1:2J_{\:1}),\:(1:J_{\:2}),\:\dots,\:(1:J_{\:K})}(z,\underline{\boldsymbol{J}})\big).
\end{align*}
\end{lemma}

\begin{proof}
From \citet[section 3.6]{Chakravorty2019}, we have a general linear-transformation ${\mathcal{T}^{(K)}_{\:\underline{\boldsymbol{J}}}}':\mathbb{R}^{2J_{\:1}}\times\dots\times\mathbb{R}^{J_{\:K}}\to\mathbb{R}^{2J}$ for $K\in\mathbb{N}^+$ and $\boldsymbol{J}\in{\mathbb{N}^+}^K$ satisfying:
\begin{align*}
\quad\boldsymbol{\mathsf{c}}_{\:(1:2J)}(z)={\mathcal{T}^{(K)}_{\:\underline{\boldsymbol{J}}}}'\big(\boldsymbol{\mathsf{c}}^{\:(1:2J_{\:1}),\:\dots,\:(1:J_{\:K})}(z,\boldsymbol{J})\big).
\end{align*}
Suppose $\mathcal{V}$ is the generic invertible transformation that vectorizes a multi-dimensional array in the natural-order of the index-set of the array. Then, we will be able to find a $2J\times 2J$ matrix ${\mathcal{A}^{(K)}_{\:\underline{\boldsymbol{J}}}}'$, such that:
\begin{align*}
\quad\boldsymbol{\mathsf{c}}_{\:(1:2J)}(z)={\mathcal{A}^{(K)}_{\:\underline{\boldsymbol{J}}}}'\cdot\mathcal{V}\big(\boldsymbol{\mathsf{c}}^{\:(1:2J_{\:1}),\:\dots,\:(1:J_{\:K})}(z,\boldsymbol{J})\big).
\end{align*}

Now observe that
\begin{align*}
\mathcal{V}\big(\boldsymbol{\mathsf{c}}^{\:0;\:(1:2J_{\:1}),\:\dots,\:(1:J_{\:K})}(z,\boldsymbol{J})\big)=\Big(1\ ,\mathcal{V}\big(\boldsymbol{\mathsf{c}}^{\:(1:2J_{\:1}),\:\dots,\:(1:J_{\:K})}(z,\boldsymbol{J})\big)\Big)
\end{align*}
Take $\mathcal{A}^{(K)}_{\:\underline{\boldsymbol{J}}}:=\text{diag}\big(1,{\mathcal{A}^{(K)}_{\:\underline{\boldsymbol{J}}}}'\big)$; since $\boldsymbol{\mathsf{c}}_{\:(0:2J)}(z)=\big(1,\boldsymbol{\mathsf{c}}_{\:(1:2J)}(z)\big)$, we must have
\begin{align*}
\quad\boldsymbol{\mathsf{c}}_{\:(0:2J)}(z)=\mathcal{A}^{(K)}_{\:\underline{\boldsymbol{J}}}\cdot\mathcal{V}\big(\boldsymbol{\mathsf{c}}^{\:0;\:(1:2J_{\:1}),\:\dots,\:(1:J_{\:K})}(z,\boldsymbol{J})\big).
\end{align*}
So, we have proved the existence of a transformation $\mathcal{T}^{(K)}_{\:\underline{\boldsymbol{J}}}$, such that
\begin{align*}
\quad\boldsymbol{\mathsf{c}}_{\:(0:2J)}(z)=\mathcal{T}^{(K)}_{\:\underline{\boldsymbol{J}}}\big(\boldsymbol{\mathsf{c}}^{\:0;\:(1:2J_{\:1}),\:\dots,\:(1:J_{\:K})}(z,\boldsymbol{J})\big).
\end{align*}
\end{proof}

Now consider a $p$-dimensional vector $\mathbf{x}\in(\mathbf{0},\mathbf{1})$, so that if $\mathbf{x}=\big(x_{\:1},\dots,x_{\:p}\big)$, then $x_{\:l}\in(0,1)$ for $1\leq l\leq p$. Let us define the $p$-dimensional function:
\begin{align*}
\boldsymbol{\mathsf{c}}_{\:(0:2J);\:p}(\mathbf{x})=\otimes_{l=1}^p\boldsymbol{\mathsf{c}}_{\:(0:2J)}(x_{\:l}).
\end{align*}
Note that $\boldsymbol{\mathsf{c}}_{\:(0:2J);\:p}(\mathbf{x})$ can be identified as a $p$-dimensional array with dimensions $(2J+1)\times\dots\times(2J+1)$ ($p$ times). This $p$-dimensional array is indexed by the index set
\begin{align*}
\mathbb{N}_{\:J;\:p}=\{0,\dots,2J\}\times\dots\times\{0,\dots,2J\}\:(p\text{ times}).
\end{align*}
For $\boldsymbol{j}=\big(j_{\:1},\dots,j_{\:p}\big)\in\mathbb{N}_{\:J;\:p}$, a general $\boldsymbol{j}$th element of $\boldsymbol{\mathsf{c}}_{\:(0:2J);\:p}(\mathbf{x})$ is the function:
\begin{align*}
\boldsymbol{\mathsf{c}}_{\:\boldsymbol{j}}(\mathbf{x})=\prod_{l=1}^p\mathsf{c}_{\:j_{\:l}}(x_{\:l}).
\end{align*}

Let us also define:
\begin{align*}
\tilde{\boldsymbol{\mathsf{c}}}^{{\:0;\:(1:2J_{\:1}),\:(1:J_{\:2}),\:\dots,\:(1:J_{\:K})}^{\:p}}\big(\mathbf{x},\underline{\boldsymbol{J}}\big)=\otimes_{l=1}^p\Big(\tilde{\boldsymbol{\mathsf{c}}}^{\:0;\:(1:2J_{\:1}),\:(1:J_{\:2}),\:\dots,\:(1:J_{\:K})}\big(x_{\:l},\underline{\boldsymbol{J}}\big)\Big).
\end{align*}
This array-like function is indexed by the index set $\mathbb{N}_{\:\underline{\boldsymbol{J}};\:p}$, defined as:
\begin{align*}
\mathbb{N}_{\:\underline{\boldsymbol{J}};\:p}=\mathbb{N}_{\:\underline{\boldsymbol{J}}}\times\dots\times\mathbb{N}_{\:\underline{\boldsymbol{J}}}\:(p\text{ times}).
\end{align*}
Let us denote a general index element of $\mathbb{N}_{\:\{\underline{\boldsymbol{J}},\:p\}}$ as $\underline{\underline{\boldsymbol{j}}}$, which is a $p$-tuple $\big(\underline{\boldsymbol{j}}_{\:1},\dots,\underline{\boldsymbol{j}}_{\:p}\big)$. For $1\leq l\leq p$, $\underline{\boldsymbol{j}}_{\:l}$ is either $0$ or a $K$-element-vector $\big(j_{\:l,\:1},\dots,j_{\:l,\:K}\big)$.

Now, the $\underline{\underline{\boldsymbol{j}}}$th element of $\tilde{\boldsymbol{\mathsf{c}}}^{{\:0,\:(1:2J_{\:1}),\:(1:J_{\:2}),\:\dots,\:(1:J_{\:K})}^{\:p}}\big(\mathbf{x},\underline{\boldsymbol{J}}\big)$ is:
\begin{align*}
\tilde{\mathsf{c}}^{\:\underline{\underline{\boldsymbol{j}}}}\big(\mathbf{x},\underline{\boldsymbol{J}}\big)=\prod_{l=1}^p\Big(\mathsf{c}^{\:j_{\:l,\:1}}(x_{\:l})\cdot \prod_{k=2}^K\dot{\mathsf{c}}^{\:j_{\:l,\:k}}(x_{\:l},L_{\:k-1})\Big)^{\mathsf{1}\big(\underline{\boldsymbol{j}}_{\:l}\neq 0\big)}.
\end{align*}

We have the following lemma:
\begin{lemma}\label{chev:lem:gen:pd}
Let $\mathbf{x}$ be an element of the $p$-dimensional open interval $(\mathbf{0},\mathbf{1})$. For $K\in\mathbb{N}^+$ and $\underline{\boldsymbol{J}}\in{\mathbb{N}^+}^K$, there exists a linear transformation $\mathcal{T}^{(K,p)}_{\:\underline{\boldsymbol{J}}}$, such that:
\begin{align*}
\boldsymbol{\mathsf{c}}_{{\:(0:2J)}^{\:p}}(\mathbf{x})=\mathcal{T}^{(K,p)}_{\:\underline{\boldsymbol{J}}}\big(\tilde{\boldsymbol{\mathsf{c}}}^{{\:0;\:(1:2J_{\:1}),\:(1:J_{\:2}),\:\dots,\:(1:J_{\:K})}^{\:p}}(\mathbf{x},\underline{\boldsymbol{J}})\big).
\end{align*}
\end{lemma}

\begin{proof}
From Lemma \ref{chev:lem:gen:1d}, we have for $1\leq l\leq p$:
\begin{align*}
\boldsymbol{\mathsf{c}}_{{\:(0:2J)}}(x_{\:l})=\mathcal{T}^{(K)}_{\:\underline{\boldsymbol{J}}}\big(\tilde{\boldsymbol{\mathsf{c}}}^{\:0;\:(1:2J_{\:1}),\:(1:J_{\:2}),\:\dots,\:(1:J_{\:K})}(x_{\:l},\underline{\boldsymbol{J}})\big).
\end{align*}
If we take tensor product of both sides for $1\leq l\leq p$, we get:
\begin{align*}
&\otimes_{l=1}^p\boldsymbol{\mathsf{c}}_{\:(0:2J)}(x_{\:l})=\otimes_{l=1}^p\Big(\mathcal{T}^{(K)}_{\:\underline{\boldsymbol{J}}}\big(\tilde{\boldsymbol{\mathsf{c}}}^{\:0;\:(1:2J_{\:1}),\:(1:J_{\:2}),\:\dots,\:(1:J_{\:K})}(x_{\:l},\underline{\boldsymbol{J}})\big)\Big) \\
&\Leftrightarrow\boldsymbol{\mathsf{c}}_{{\:(0:2J)}^{\:p}}(\mathbf{x})=\big(\otimes_{l=1}^p\mathcal{T}^{(K)}_{\:\underline{\boldsymbol{J}}}\big)\Big(\otimes_{l=1}^p\big(\tilde{\boldsymbol{\mathsf{c}}}^{\:0;\:(1:2J_{\:1}),\:(1:J_{\:2}),\:\dots,\:(1:J_{\:K})}(x_{\:l},\underline{\boldsymbol{J}})\big)\Big) \\
&\Leftrightarrow\boldsymbol{\mathsf{c}}_{{\:(0:2J)}^{\:p}}(\mathbf{x})=\mathcal{T}^{(K,p)}_{\:\underline{\boldsymbol{J}}}\big(\tilde{\boldsymbol{\mathsf{c}}}^{{\:0;\:(1:2J_{\:1}),\:(1:J_{\:2}),\:\dots,\:(1:J_{\:K})}^{\:p}}(\mathbf{x},\underline{\boldsymbol{J}})\big).
\end{align*}
If we let $\mathcal{T}^{(K,p)}_{\:\underline{\boldsymbol{J}}}:=\otimes_{l=1}^p\mathcal{T}^{(K)}_{\:\underline{\boldsymbol{J}}}$, we have:
\begin{align*}
\boldsymbol{\mathsf{c}}_{{\:(0:2J)}^{\:p}}(\mathbf{x})=\mathcal{T}^{(K,p)}_{\:\underline{\boldsymbol{J}}}\big(\tilde{\boldsymbol{\mathsf{c}}}^{{\:0;\:(1:2J_{\:1}),\:(1:J_{\:2}),\:\dots,\:(1:J_{\:K})}^{\:p}}(\mathbf{x},\underline{\boldsymbol{J}})\big).
\end{align*}
\end{proof}

Now remember, we have the following for $\boldsymbol{j}\in\mathbb{N}_{\:J;\:p}$:
\begin{align*}
\boldsymbol{\mathsf{C}}_{\:\boldsymbol{j}}(\mathbf{X})=\sum_{i\in I}\boldsymbol{\mathsf{c}}_{\:\boldsymbol{j}}(\mathbf{x}_{\:i}),\quad\bar{\boldsymbol{\mathsf{C}}}_{\:\boldsymbol{j}}(\mathbf{X})=\big(\nicefrac{1}{\mathbf{X}}\big)\cdot\boldsymbol{\mathsf{C}}_{\:\boldsymbol{j}}(\mathbf{X}).
\end{align*}
Let us define the collections
\begin{align*}
\boldsymbol{\mathsf{C}}_{{\:(0:2J)}^{\:p}}(\mathbf{X})=\{\{\boldsymbol{\mathsf{C}}_{\:\boldsymbol{j}}(\mathbf{X})\}\}_{\:\boldsymbol{j}\in\mathbb{N}_{\:J;\:p}}, \quad
\bar{\boldsymbol{\mathsf{C}}}_{{\:(0:2J)}^{\:p}}(\mathbf{X})=\{\{\bar{\boldsymbol{\mathsf{C}}}_{\:\boldsymbol{j}}(\mathbf{X})\}\}_{\:\boldsymbol{j}\in\mathbb{N}_{\:J;\:p}},
\end{align*}
and note that
\begin{align*}
\boldsymbol{\mathsf{C}}_{{\:(0:2J)}^{\:p}}(\mathbf{X})=\sum_{i\in I}\boldsymbol{\mathsf{c}}_{{\:(0:2J)}^{\:p}}(\mathbf{x}_{\:i}),\quad\bar{\boldsymbol{\mathsf{C}}}_{{\:(0:2J)}^{\:p}}(\mathbf{X})=\big(\nicefrac{1}{\mathbf{X}}\big)\cdot\boldsymbol{\mathsf{C}}_{{\:(0:2J)}^{\:p}}(\mathbf{X}).
\end{align*}

We then define the following for a general $\underline{\underline{\boldsymbol{j}}}\in\mathbb{N}_{\:\underline{\boldsymbol{J}};\:p}$:
\begin{align*}
\tilde{\boldsymbol{\mathsf{C}}}^{\:\underline{\underline{\boldsymbol{j}}}}(\mathbf{X},\boldsymbol{J})=\sum_{i\in I}\tilde{\boldsymbol{\mathsf{c}}}^{\:\underline{\underline{\boldsymbol{j}}}}(\mathbf{x}_{\:i},\boldsymbol{J}),\quad\bar{\tilde{\boldsymbol{\mathsf{C}}}}^{\:\underline{\underline{\boldsymbol{j}}}}(\mathbf{X},\boldsymbol{J})=\big(\nicefrac{1}{\mathbf{X}}\big)\cdot\tilde{\boldsymbol{\mathsf{C}}}^{\:\underline{\underline{\boldsymbol{j}}}}(\mathbf{X},\boldsymbol{J}).
\end{align*}
Again, we define the collections of statistics: 
\begin{align*}
&\tilde{\boldsymbol{\mathsf{C}}}^{{\:0;\:(1:2J_{\:1}),\:(1:J_{\:2}),\:\dots,\:(1:J_{\:K})}^{\:p}}(\mathbf{X},\underline{\boldsymbol{J}})=\{\{\tilde{\boldsymbol{\mathsf{C}}}^{\:\underline{\underline{\boldsymbol{j}}}}(\mathbf{X},\boldsymbol{J})\}\}_{\:\underline{\underline{\boldsymbol{j}}}\in\mathbb{N}_{\:\underline{\boldsymbol{J}};\:p}} \\
&\bar{\tilde{\boldsymbol{\mathsf{C}}}}^{{\:0;\:(1:2J_{\:1}),\:(1:J_{\:2}),\:\dots,\:(1:J_{\:K})}^{\:p}}(\mathbf{X},\underline{\boldsymbol{J}})=\{\{\bar{\tilde{\boldsymbol{\mathsf{C}}}}^{\:\underline{\underline{\boldsymbol{j}}}}(\mathbf{X},\boldsymbol{J})\}\}_{\:\underline{\underline{\boldsymbol{j}}}\in\mathbb{N}_{\:\underline{\boldsymbol{J}};\:p}}.
\end{align*}
Furthermore, we have the identities
\begin{align*}
&\tilde{\boldsymbol{\mathsf{C}}}^{{\:0;\:(1:2J_{\:1}),\:(1:J_{\:2}),\:\dots,\:(1:J_{\:K})}^{\:p}}(\mathbf{X},\underline{\boldsymbol{J}})=\sum_{i\in I}\tilde{\boldsymbol{\mathsf{c}}}^{{\:0;\:(1:2J_{\:1}),\:(1:J_{\:2}),\:\dots,\:(1:J_{\:K})}^{\:p}}(\mathbf{x}_{\:i},\underline{\boldsymbol{J}}), \\
&\hspace{-1.2em}\bar{\tilde{\boldsymbol{\mathsf{C}}}}^{{\:0;\:(1:2J_{\:1}),\:(1:J_{\:2}),\:\dots,\:(1:J_{\:K})}^{\:p}}(\mathbf{X},\underline{\boldsymbol{J}})=\big(\nicefrac{1}{\mathbf{X}}\big)\cdot\tilde{\boldsymbol{\mathsf{C}}}^{{\:0;\:(1:2J_{\:1}),\:(1:J_{\:2}),\:\dots,\:(1:J_{\:K})}^{\:p}}(\mathbf{X},\underline{\boldsymbol{J}}).
\end{align*}

Now we are ready to prove Theorem \ref{chev:thm:main}.

\begin{proof}
From Lemma \ref{chev:lem:gen:pd}, we have for $i\in I$:
\begin{align*}
\boldsymbol{\mathsf{c}}_{{\:(0:2J)}^{\:p}}(\mathbf{x}_{\:i})=\mathcal{T}^{(K,p)}_{\:\underline{\boldsymbol{J}}}\big(\tilde{\boldsymbol{\mathsf{c}}}^{{\:0;\:(1:2J_{\:1}),\:(1:J_{\:2}),\:\dots,\:(1:J_{\:K})}^{\:p}}(\mathbf{x}_{\:i},\underline{\boldsymbol{J}})\big).
\end{align*}
Since $\mathcal{T}^{(K,p)}_{\:\underline{\boldsymbol{J}}}$ is a linear transformation, if we take a sum over $i\in I$ and divide both sides by $|\mathbf{X}|$, we obtain the desired result:
\begin{align*}
&\hspace{-1em}\frac{\sum_{i\in I}\boldsymbol{\mathsf{c}}_{{\:(0:2J)}^{\:p}}(\mathbf{x}_{\:i})}{|\mathbf{X}|}=\frac{\sum_{i\in I}\mathcal{T}^{(K,p)}_{\:\underline{\boldsymbol{J}}}\big(\tilde{\boldsymbol{\mathsf{c}}}^{{\:0;\:(1:2J_{\:1}),\:(1:J_{\:2}),\:\dots,\:(1:J_{\:K})}^{\:p}}(\mathbf{x}_{\:i},\underline{\boldsymbol{J}})\big)}{|\mathbf{X}|}\\
&\Leftrightarrow\frac{\sum_{i\in I}\boldsymbol{\mathsf{c}}_{{\:(0:2J)}^{\:p}}(\mathbf{x}_{\:i})}{|\mathbf{X}|}=\mathcal{T}^{(K,p)}_{\:\underline{\boldsymbol{J}}}\bigg(\frac{\sum_{i\in I}\tilde{\boldsymbol{\mathsf{c}}}^{{\:0;\:(1:2J_{\:1}),\:(1:J_{\:2}),\:\dots,\:(1:J_{\:K})}^{\:p}}(\mathbf{x}_{\:i},\underline{\boldsymbol{J}})}{|\mathbf{X}|}\bigg) \\
&\Leftrightarrow\bar{\boldsymbol{\mathsf{C}}}_{{\:(0:2J)}^{\:p}}(X)=\mathcal{T}^{(K,p)}_{\:\underline{\boldsymbol{J}}}\Big(\bar{\tilde{\boldsymbol{\mathsf{C}}}}^{{\:0;\:(1:2J_{\:1}),\:(1:J_{\:2}),\:\dots,\:(1:J_{\:K})}^{\:p}}(X,\underline{\boldsymbol{J}})\Big).
\end{align*}
\end{proof}

\end{document}